\def\year{2016}\relax
\documentclass[letterpaper]{article}

\pdfoutput=1
\usepackage{aaai16}
\usepackage{times}
\usepackage{helvet}
\usepackage{courier}

\usepackage{amsmath}
\usepackage{amssymb} 

\newcommand{\ms}[1]{\mbox{\scriptsize #1}}
\DeclareMathOperator*{\argmin}{arg\,min}

\newtheorem{invariant}{Invariant}

\newtheorem{theorem}{Theorem}
\newenvironment{proof}[1][Proof]{\begin{trivlist}
   \item[\hskip \labelsep {\bfseries #1}]}{\hfill$\square$\end{trivlist}}

\usepackage{booktabs}

\usepackage{caption}
\usepackage{subcaption}
\usepackage{tikz}
\usetikzlibrary{backgrounds}
\usepackage{pgfplots}
\usepackage{algorithm}
\usepackage[noend]{algpseudocode}

\ifx\aaaiversion\undefined
\usepackage{balance}
\fi

\makeatletter
\pgfdeclareshape{document}{
\inheritsavedanchors[from=rectangle] 
\inheritanchorborder[from=rectangle]
\inheritanchor[from=rectangle]{center}
\inheritanchor[from=rectangle]{north}
\inheritanchor[from=rectangle]{south}
\inheritanchor[from=rectangle]{west}
\inheritanchor[from=rectangle]{east}
\backgroundpath{
\southwest \pgf@xa=\pgf@x \pgf@ya=\pgf@y
\northeast \pgf@xb=\pgf@x \pgf@yb=\pgf@y
\pgf@xc=\pgf@xb \advance\pgf@xc by-4pt 
\pgf@yc=\pgf@yb \advance\pgf@yc by-4pt
\pgfpathmoveto{\pgfpoint{\pgf@xa}{\pgf@ya}}
\pgfpathlineto{\pgfpoint{\pgf@xa}{\pgf@yb}}
\pgfpathlineto{\pgfpoint{\pgf@xc}{\pgf@yb}}
\pgfpathlineto{\pgfpoint{\pgf@xb}{\pgf@yc}}
\pgfpathlineto{\pgfpoint{\pgf@xb}{\pgf@ya}}
\pgfpathclose
\pgfpathmoveto{\pgfpoint{\pgf@xc}{\pgf@yb}}
\pgfpathlineto{\pgfpoint{\pgf@xc}{\pgf@yc}}
\pgfpathlineto{\pgfpoint{\pgf@xb}{\pgf@yc}}
\pgfpathlineto{\pgfpoint{\pgf@xc}{\pgf@yc}}
}
}
\makeatother

\title{
A Unifying Formalism for Shortest Path Problems with Expensive\\
Edge Evaluations via Lazy Best-First Search over Paths with Edge Selectors
}
\author{
   Christopher M.~Dellin \and Siddhartha S.~Srinivasa\\
   The Robotics Institute, Carnegie Mellon University\\
      \{cdellin, siddh\}@cs.cmu.edu\\
}

\ifx\aaaiversion\undefined
\nocopyright
\fi

\begin{document}

\maketitle

\ifx\aaaiversion\undefined
\insert\footins{\noindent\footnotesize%
Presented at ICAPS 2016, London. This extended version includes
proofs and timing results in the appendix.}
\fi

\begin{abstract}
While the shortest path problem has myriad applications,
the computational efficiency of suitable algorithms
depends intimately on the underlying problem domain.
In this paper,
we focus on domains where evaluating the edge weight function
dominates algorithm running time.
Inspired by approaches in robotic motion planning,
we define and investigate the \emph{Lazy Shortest Path} class of
algorithms which is differentiated by the choice of
an \emph{edge selector} function.
We show that several algorithms in the literature are equivalent to
this lazy algorithm for appropriate choice of this selector.
Further, we propose various novel selectors inspired by
sampling and statistical mechanics,
and find that these selectors outperform
existing algorithms on a set of example problems.
\end{abstract}

\section{Introduction}

Graphs provide a powerful abstraction
capable of representing problems in a wide variety of domains
from computer networking to puzzle solving
to robotic motion planning.
In particular,
many important problems can be captured
as \emph{shortest path problems} (Figure~\ref{fig:sp-intro}),
wherein a path $p^*$ of minimal length is desired
between two query vertices through a graph $G$
with respect to an edge weight function $w$.

Despite the expansive applicability of this single abstraction,
there exist a wide variety of algorithms in the literature
for solving the shortest path problem efficiently.
This is because the measure of computational efficiency,
and therefore the correct choice of algorithm,
is inextricably tied to the underlying problem domain.

The computational costs incurred by an algorithm
can be broadly categorized into three sources
corresponding to the blocks in Figure~\ref{fig:sp-intro}.
One such source consists of queries on the structure
of the graph $G$ itself.
The most commonly discussed such operation,
\emph{expanding} a vertex (determining its successors),
is especially fundamental
when the graph is represented implicitly,
e.g. for domains with large graphs
such as the 15-puzzle or Rubik's cube.
It is with respect to vertex expansions
that A* \cite{hart1968astar} is optimally efficient.

A second source of computational cost consists of maintaining
ordered data structures inside the algorithm itself,
which is especially important for problems with large branching
factors.
For such domains,
approaches such as partial expansion \cite{yoshizumi2000peastar}
or iterative deepening \cite{korf1985idastar}
significantly reduce the number of vertices generated and stored
by either selectively filtering surplus vertices from the frontier,
or by not storing the frontier at all.

The third source of computational cost arises not from reasoning
over the structure of $G$,
but instead from evaluating the edge weight function $w$
(i.e. we treat discovering an out-edge and determining its weight
separately).
Consider for example the problem of articulated robotic motion planning
using roadmap methods \cite{kavrakietal1996prm}.
While these graphs are often quite small
(fewer than $10^5$ vertices),
determining the weight of each edge requires performing many
collision and distance computations for the complex geometry
of the robot and environment,
resulting in planning times of multiple seconds to find a path.

In this paper,
we consider problem domains in which evaluating the edge weight
function $w$ dominates algorithm running time
and investigate the following research question:
\begin{quote}
How can we minimize the number of edges we need to evaluate
to answer shortest-path queries?
\end{quote}

\begin{figure}
\centering
\begin{tikzpicture}
   \tikzset{>=latex} 
   \node[draw,align=center,minimum height=1.0cm,thick]
      at (2.5,2.5) {Graph\\$G=(V,E)$};
   \node[draw,align=center,minimum height=1.0cm,thick]
      at (5.5,2.5) {Weight Function\\$w:E \rightarrow [0,+\infty]$};
   \node[draw,align=center,minimum height=1.0cm,minimum width=3cm,thick]
      (alg) at (4,1) {Shortest Path\\Algorithm};
   \node[draw,align=center,shape=document,minimum width=1.5cm,ultra thin]
      (query) at (1,1) {Query $u$};
   \node[draw,align=center,shape=document,minimum width=1.5cm,ultra thin]
      (path) at (7,1) {Path $p^*$};
   \draw[->] (3,2) -- (3,1.5);
   \draw[->] (5,2) -- (5,1.5);
   \draw[->] (query.east) -- (alg.west);
   \draw[->] (alg.east) -- (path.west);
\end{tikzpicture}
\caption{While solving a shortest path query,
   a shortest path algorithm incurs computation cost from three sources:
   examining the structure of the graph $G$,
   evaluating the edge weight function $w$,
   and maintaining internal data structures.}
\label{fig:sp-intro}
\end{figure}
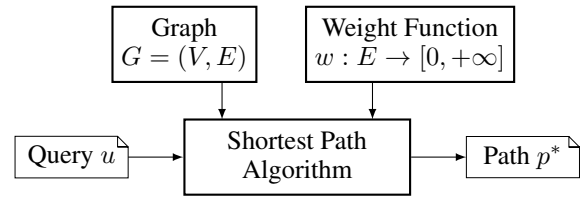

We make three primary contributions.
First,
inspired by lazy collision checking techniques from 
robotic motion planning \cite{bohlin2000lazyprm},
we formulate a class of shortest-path algorithms 
that is well-suited to problem domains with expensive edge evaluations.
Second,
we show that several existing algorithms in the literature
can be expressed as special cases of this algorithm.
Third,
we show that the extensibility afforded by the algorithm allows for
novel edge evaluation strategies,
which can outperform existing algorithms
over a set of example problems.

\section{Lazy Shortest Path Algorithm}

We describe a lazy approach to finding short paths
which is well-suited to domains with
expensive edge evaluations.

\subsection{Problem Definition}

A path $p$ in a graph $G = (V,E)$
is composed of a sequence of adjacent edges 
connecting two endpoint vertices.
Given an edge weight function
$w : E \rightarrow [0,+\infty]$,
the length of the path with respect to $w$ is then:
\begin{equation}
   \mbox{len}(p, w) = \sum_{e \in p} w(e).
\end{equation}
Given a single-pair planning query
$u: (v_{\ms{start}}, v_{\ms{goal}})$
inducing a set of satisfying paths $P_u$,
the \emph{shortest-path problem} is:
\begin{equation}
   p^* = \argmin_{p \, \in \, P_u} \mbox{len}(p, w).
   \label{eqn:objective}
\end{equation}

A shortest-path algorithm computes $p^*$
given $(G, u, w)$.
Many such algorithms have been proposed
to efficiently accommodate a wide array of underlying problem domains.
The well-known principle of best-first search (BFS)
is commonly employed to select vertices for expansion
so as to minimize such expansions while guaranteeing optimality.
Since we seek to minimize edge evaluations,
we apply BFS to the question of selecting candidate paths in
$G$ for evaluation.
The resulting algorithm, Lazy Shortest Path (LazySP),
is presented in Algorithm~\ref{alg:lazy-outline},
and can be applied to graphs defined implicitly or explicitly.

\subsection{The Algorithm}

\begin{algorithm}[t]
\caption{Lazy Shortest Path (LazySP)}
\label{alg:lazy-outline}
\begin{algorithmic}[1]
\Function {\textsc{LazyShortestPath}}{$G, u, w, w_{\ms{est}}$}
\State $E_{\ms{eval}} \leftarrow \emptyset$ 
\State $w_{\ms{lazy}}(e) \leftarrow w_{\ms{est}}(e) \quad \forall e \in E$ 
\Loop
   \State $p_{\ms{candidate}} \leftarrow
      \mbox{\sc ShortestPath}(G, u, w_{\ms{lazy}})$ 
      \label{line:lazy-outline-shortestpath}
   \If {$p_{\ms{candidate}} \subseteq E_{\ms{eval}}$} 
      \State \Return $p_{\ms{candidate}}$ 
   \EndIf
   \State $E_{\ms{selected}} \leftarrow  \mbox{\sc Selector}(G, p_{\ms{candidate}})$ 
   \label{line:lazy-outline-chooseedges} 
   \For {$e \in E_{\ms{selected}} \setminus E_{\ms{eval}}$} 
      \State $w_{\ms{lazy}}(e) \leftarrow w(e)$ \Comment Evaluate (expensive)
      \State $E_{\ms{eval}} \leftarrow E_{\ms{eval}} \cup e$ 
   \EndFor
\EndLoop
\EndFunction
\end{algorithmic}
\end{algorithm}

We track evaluated edges with the set $E_{\ms{eval}}$.
We are given an estimator function $w_{\ms{est}}$ of the true edge weight $w$.
This estimator is inexpensive to compute
(e.g. edge length or even $0$).
We then define a \emph{lazy} weight function $w_{\ms{lazy}}$
which returns the
true weight of an evaluated edge and otherwise
uses the inexpensive estimator $w_{\ms{est}}$.

At each iteration of the search,
the algorithm uses $w_{\ms{lazy}}$ to compute a candidate path
$p_{\ms{candidate}}$
by calling an existing solver \textsc{ShortestPath}
(note that this invocation requires no evaluations of $w$).
Once a candidate path has been found,
it is returned if it is fully evaluated.
Otherwise,
an \emph{edge selector} is employed which selects
graph edge(s) for evaluation.
The true weights of these edges are then evaluated
(incurring the requisite computational cost),
and the algorithm repeats.


LazySP is complete and optimal:

\begin{theorem}[Completeness of LazySP]
If the graph $G$ is finite,
\textsc{ShortestPath} is complete,
and the set $E_{\ms{selected}}$
returned by \textsc{Selector}
returns at least one unevaluated edge on $p_{\ms{candidate}}$,
then \textsc{LazyShortestPath} is complete.
\label{thm:lazy-completeness}
\end{theorem}

\begin{theorem}[Optimality of LazySP]
If $w_{\ms{est}}$ is chosen such that
$w_{\ms{est}}(e) \leq \epsilon \, w(e)$ for some parameter
$\epsilon \geq 1$ and
\textsc{LazyShortestPath} terminates
with some path $p_{\ms{ret}}$,
then $\mbox{len}(p_{\ms{ret}}, w) \leq \epsilon \, \ell^*$
with $\ell^*$ the length of an optimal path.
\label{thm:lazy-optimality}
\end{theorem}

The optimality of LazySP depends on the admissibility of
$w_{\ms{est}}$
in the same way that the optimality of A* depends on
the admissibility of its goal heuristic $h$.
Theorem~\ref{thm:lazy-optimality} establishes the general
bounded suboptimality of LazySP
w.r.t. the inflation parameter $\epsilon$.
While our theoretical results (e.g. equivalences)
hold for any choice of $\epsilon$,
for clarity our examples and experimental results
focus on cases with $\epsilon = 1$.
\ifx\aaaiversion\undefined
Proofs are available in the appendix.
\else
Proofs are available in \cite{dellin2016lazyspextended}.
\fi

\subsection{The Edge Selector: Key to Efficiency}

\begin{algorithm}[t]
\caption{Various Simple LazySP Edge Selectors}
\begin{algorithmic}[1]
\Function {\textsc{SelectExpand}}{$G, p_{\ms{candidate}}$}
   \State $e_{\ms{first}} \leftarrow$ first unevaluated $e \in p_{\ms{candidate}}$
   \State $v_{\ms{frontier}} \leftarrow G.\mbox{source}(e_{\ms{first}})$
   \State $E_{\ms{selected}} \leftarrow G.\mbox{out\_edges}(v_{\ms{frontier}})$
   \State \Return $E_{\ms{selected}}$
\EndFunction
\vspace{0.02in}
\Function {\textsc{SelectForward}}{$G, p_{\ms{candidate}}$}
   \State \Return $\{ \mbox{first unevaluated } e \in p_{\ms{candidate}} \}$
\EndFunction
\vspace{0.02in}
\Function {\textsc{SelectReverse}}{$G, p_{\ms{candidate}}$}
   \State \Return $\{ \mbox{last unevaluated } e \in p_{\ms{candidate}} \}$
\EndFunction
\vspace{0.02in}
\Function {\textsc{SelectAlternate}}{$G, p_{\ms{candidate}}$}
   \If {LazySP iteration number is odd}
      \State \Return $\{ \mbox{first unevaluated } e \in p_{\ms{candidate}} \}$
   \Else
      \State \Return $\{ \mbox{last unevaluated } e \in p_{\ms{candidate}} \}$
   \EndIf
\EndFunction
\vspace{0.02in}
\Function {\textsc{SelectBisection}}{$G, p_{\ms{candidate}}$}
   \State \Return $\left\{ \begin{array}{ll}
      \mbox{unevaluated } e \in p_{\ms{candidate}} \\
      \mbox{furthest from nearest evaluated edge}
      \end{array} \right\}$
\EndFunction
\end{algorithmic}
\label{alg:simple-selectors}
\end{algorithm}

The LazySP algorithm exhibits a rough similarity to optimal
replanning algorithms such as
D* \cite{stentz1994dstar}
which plan a sequence of shortest paths for a mobile robot
as new edge weights are discovered during its traverse.
D* treats edge changes
passively as an aspect of the problem setting
(e.g. a sensor with limited range).

The key difference is that our problem setting treats 
edge evaluations as an active choice that can be exploited.
While any choice of edge selector that meets the conditions above
will lead to an algorithm that is complete and optimal,
its \emph{efficiency} is dictated by the choice of this
selector.
This motivates the theoretical and empirical investigation of different
edge selectors in this paper.

\begin{figure*}[t!]
   \centering
   \begin{subfigure}[b]{2.47cm}
      \centering
      \includegraphics{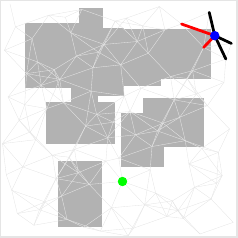} \\
      \vspace{0.04in}
      \includegraphics{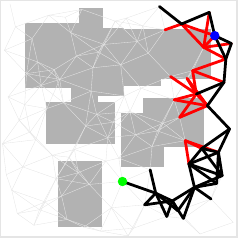} \\
      \vspace{0.04in}
      \includegraphics{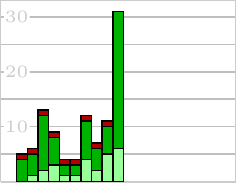} \\
      \caption{Expand[77]}
   \end{subfigure}
   \begin{subfigure}[b]{2.47cm}
      \centering
      \includegraphics{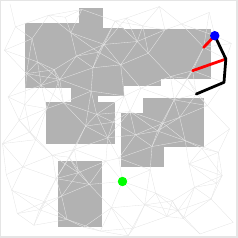} \\
      \vspace{0.04in}
      \includegraphics{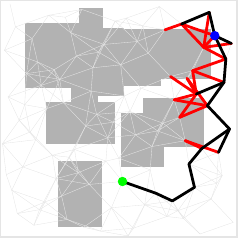} \\
      \vspace{0.04in}
      \includegraphics{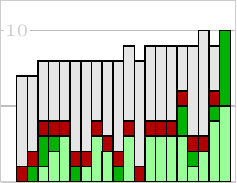} \\
      \caption{Forward[34]}
   \end{subfigure}
   \begin{subfigure}[b]{2.47cm}
      \centering
      \includegraphics{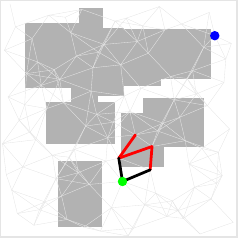} \\
      \vspace{0.04in}
      \includegraphics{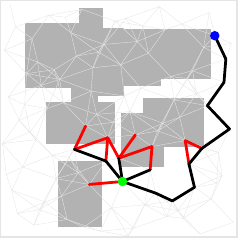} \\
      \vspace{0.04in}
      \includegraphics{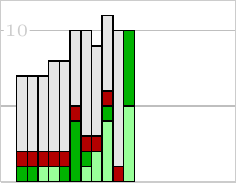} \\
      \caption{Reverse[24]}
   \end{subfigure}
   \begin{subfigure}[b]{2.47cm}
      \centering
      \includegraphics{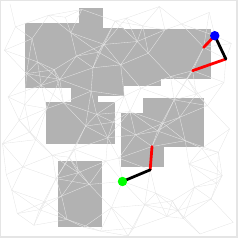} \\
      \vspace{0.04in}
      \includegraphics{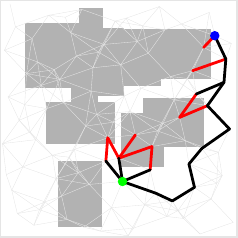} \\
      \vspace{0.04in}
      \includegraphics{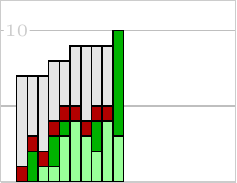} \\
      \caption{Alternate[23]}
   \end{subfigure}
   \begin{subfigure}[b]{2.47cm}
      \centering
      \includegraphics{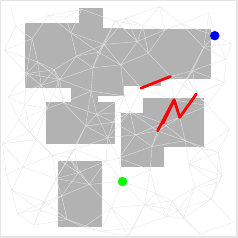} \\
      \vspace{0.04in}
      \includegraphics{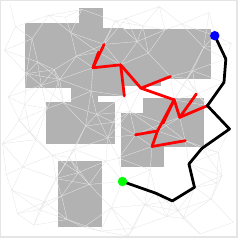} \\
      \vspace{0.04in}
      \includegraphics{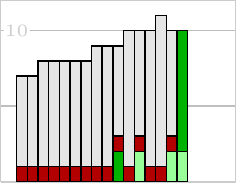} \\
      \caption{Bisection[25]}
   \end{subfigure}
   \begin{subfigure}[b]{2.47cm}
      \centering
      \includegraphics{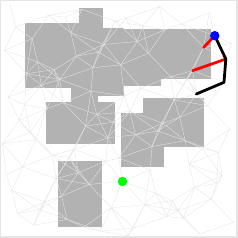} \\
      \vspace{0.04in}
      \includegraphics{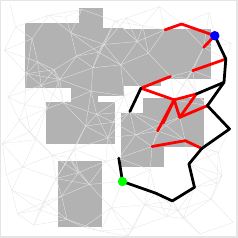} \\
      \vspace{0.04in}
      \includegraphics{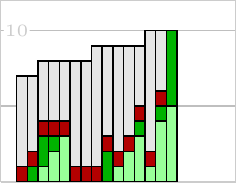} \\
      \caption{\scalebox{0.97}[1.0]{WeightSamp[22]}}
   \end{subfigure}
   \begin{subfigure}[b]{2.47cm}
      \centering
      \includegraphics{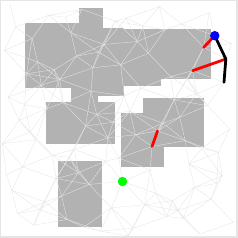} \\
      \vspace{0.04in}
      \includegraphics{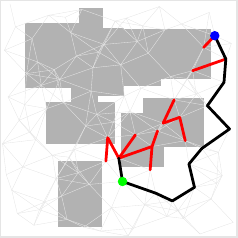} \\
      \vspace{0.04in}
      \includegraphics{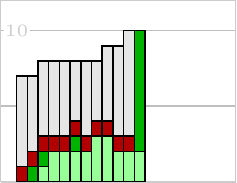} \\
      \caption{Partition[22]}
   \end{subfigure}
   \caption{Snapshots of the LazySP algorithm using each edge selector
      discussed in this paper on the same obstacle roadmap graph problem,
      with start (\tikz[baseline=-0.5ex]{\node[circle,fill=blue,inner sep=1pt]{};})
      and goal (\tikz[baseline=-0.5ex]{\node[circle,fill=green,inner sep=1pt]{};}).
      At top, the algorithms after evaluating five edges
      (evaluated edges labeled as
      \tikz{\draw[very thick] (0,0) -- (0.15,0.15);}  valid
      or \tikz{\draw[very thick,red] (0,0) -- (0.15,0.15);} invalid).
      At middle, the final set of evaluated edges.
      At bottom, for each unique path considered from left to right,
      the number of edges on the path that are
      \tikz{\node[fill=green!40!white,draw=black]{};}\;already evaluated,
      \tikz{\node[fill=green!70!black,draw=black]{};}\;evaluated and valid,
      \tikz{\node[fill=red!70!black,draw=black]{};}\;evaluated and invalid,
      and \tikz{\node[fill=black!10!white,draw=black]{};}\;unevaluated.
      The total number of edges evaluated is noted in brackets.
      Note that the scale on the Expand plot has been adjusted
      because the selector evaluates many edges not on the candidate
      path at each iteration.
      }
   \label{fig:snapshots}
\end{figure*}

\subsubsection{Simple selectors.}
We codify five common strategies in
Algorithm~\ref{alg:simple-selectors}.
The Expand selector captures the edge weights that are evaluated
during a conventional vertex expansion.
The selector identifies the first unevaluated edge
$e_{\ms{first}}$ on the candidate path,
and considers the source vertex of this edge a \emph{frontier} vertex.
It then selects all out-edges of this frontier vertex
for evaluation.
The Forward and Reverse selectors select the first and last
unevaluated edge on the candidate path, respectively
(note that Forward returns a subset of Expand).

The Alternate selector simply alternates between Forward
and Reverse on each iteration.
This can be motivated by both bidirectional search algorithms
as well as motion planning algorithms such as
RRT-Connect \cite{kuffner2000rrtconnect}
which tend to perform well w.r.t. state evaluations.

The Bisection selector
chooses among those unevaluated edges
the one furthest from an evaluated edge on the candidate path.
This selector is roughly analogous to the collision checking strategy
employed by the Lazy PRM \cite{bohlin2000lazyprm}
as applied to our problem on abstract graphs.

In the following section,
we demonstrate that instances of LazySP using simple selectors
yield equivalent results to existing vertex algorithms.
We then discuss two more sophisticated
selectors motivated by weight function sampling
and statistical mechanics.

\section{Edge Equivalence to A* Variants}

\begin{table}
   \centering
   \begin{tabular}{lll}
      \toprule
      LazySP & Existing & \\
      Selector & Algorithm & Result \\
      \midrule
      Expand & (Weighted) A* & Edge-equivalent \\
      & & (Theorems \ref{thm:astar-equiv-from-lazy},
                 \ref{thm:astar-equiv-to-lazy}) \\
      \addlinespace[0.3em]
      Forward & Lazy Weighted A* & Edge-equivalent \\
      & & (Theorems \ref{thm:lwastar-equiv-from-lazy},
                 \ref{thm:lwastar-equiv-to-lazy}) \\
      \addlinespace[0.3em]
      Alternate & Bidirectional Heuristic & Conjectured \\
      & Front-to-Front Algorithm & \\
      \bottomrule
   \end{tabular}%
   \caption{LazySP equivalence results.
      The A*, LWA*, and BHFFA algorithms use reopening and the dynamic
      $h_{\ms{lazy}}$ heuristic (\ref{eqn:h_lazy}).}
   \label{table:equivalences}
\end{table}

In the previous section,
we introduced LazySP as the path-selection analogue
to BFS vertex-selection algorithms.
In this section,
we make this analogy more precise.
In particular,
we show that LazySP-Expand
is edge-equivalent to a variant of A*
(and Weighted A*),
and that LazySP-Forward is edge-equivalent to a variant of
Lazy Weighted A*
(see Table~\ref{table:equivalences}).
It is important to be specific about the conditions under which
these equivalences arise,
which we detail here.
\ifx\aaaiversion\undefined
Proofs are available in the appendix.
\else
Proofs are available in \cite{dellin2016lazyspextended}.
\fi

\subsubsection{Edge equivalence.}
We say that two algorithms are \emph{edge-equivalent} if they
evaluate the same edges in the same order.
We consider an algorithm to have evaluated an edge
the first time the edge's true weight is requested.

\subsubsection{Arbitrary tiebreaking.}
For some graphs,
an algorithm may have multiple allowable choices at each iteration
(e.g. LazySP with multiple candidate shortest paths,
or A* with multiple vertices in OPEN with lowest $f$-value).
We will say that algorithm A is equivalent to algorithm B
if for any choice available to A,
there exists an allowable choice available to B
such that the same edge(s) are evaluated by each.

\subsubsection{A* with reopening.}
We show equivalence to variants of A* and Lazy Weighted A*
that do not use a CLOSED list to prevent
vertices from being visited more than once.

\subsubsection{A* with a dynamic heuristic.}
In order to apply A* and Lazy Weighted A* to our problem,
we need a goal heuristic over vertices.
The most simple may be
\begin{equation}
   h_{\ms{est}}(v) = \min_{p : v \rightarrow v_g} \mbox{len}(p, w_{\ms{est}}).
   \label{eqn:h_est}
\end{equation}
Note that the value of this heuristic could be computed as a
pre-processing step using Dijkstra's algorithm \cite{dijkstra1959anote}
before iterations begin.
However,
in order for the equivalences to hold,
we require the use of the lazy heuristic
\begin{equation}
   h_{\ms{lazy}}(v) = \min_{p : v \rightarrow v_g} \mbox{len}(p, w_{\ms{lazy}}).
   \label{eqn:h_lazy}
\end{equation}
This heuristic is dynamic in that it depends on $w_{\ms{lazy}}$
which changes as edges are evaluated.
Therefore,
heuristic values must be recomputed for all affected vertices on OPEN
after each iteration.

\subsection{Equivalence to A*}

We show that the LazySP-Expand algorithm
is edge-equivalent to a variant of the A*
shortest-path algorithm.
We make use of two invariants that are maintained during the
progression of A*.
\begin{invariant}
If $v$ is discovered by A* and $v'$ is undiscovered,
with $v'$ a successor of $v$,
then $v$ is on OPEN.%
\label{inv:astar-cundisc-popen}%
\end{invariant}
\begin{invariant}
If $v$ and $v'$ are discovered by A*,
with $v'$ a successor of $v$,
and $g[v] + w(v,v') < g[v']$,
then $v$ is on OPEN.%
\label{inv:astar-wless-popen}%
\end{invariant}
When we say a vertex is \emph{discovered},
we mean that it is either on OPEN or CLOSED.
Note that Invariant \ref{inv:astar-wless-popen} holds
because we allow vertices to be reopened;
without reopening (and with an inconsistent heuristic),
later finding a cheaper path to $v$ (and not reopening $v'$)
would invalidate the invariant.

We will use the goal heuristic $h_{\ms{lazy}}$ from (\ref{eqn:h_lazy}).
Note that if an admissible edge weight estimator $\hat{w}$ exists
(that is, $\hat{w} \leq w$),
then our A* can approximate the Weighted A* algorithm
\cite{pohl1970weightedastar}
with parameter $\epsilon$
by using $w_{\ms{est}} = \epsilon \, \hat{w}$,
and the suboptimality bound from
Theorem~\ref{thm:lazy-optimality} holds.

\begin{figure}[t]
   \centering
   \begin{tikzpicture}
      
      \draw[fill=black!05] (2.1,1.5) ellipse (0.4cm and 0.5cm);
      \draw[fill=black!05] (5.9,1.5) ellipse (0.4cm and 0.5cm);
      \draw[fill=black!05] (4,1) ellipse (0.4cm and 0.4cm);
      
      \node[align=center] at (0.75,1.5) {$P_{\ms{candidate}}$\\(LazySP)};
      \node[align=center] at (7.25,1.5) {$S_{\ms{candidate}}$\\(A*)};
      \node[align=center] at (4,1.75) {$V_{\ms{frontier}}$};
      
      \node[fill=black,circle,inner sep=1pt] (p1) at (2.15,1.7) {};
      \node[fill=black,circle,inner sep=1pt] (p2) at (2.05,1.3) {};
      
      \node[fill=black,circle,inner sep=1pt] (s1) at (5.95,1.8) {};
      \node[fill=black,circle,inner sep=1pt] (s2) at (5.85,1.3) {};
      
      \node[fill=black,circle,inner sep=1pt] (v1) at (4.1,0.9) {};
      
      \draw[->] (p1) -- (v1);
      \draw[->] (p2) -- (v1);
      
      \draw[->] (s1) -- (v1);
      \draw[->] (s2) -- (v1);
      
   \end{tikzpicture}
   \caption{Illustration of the equivalence
      between A* and LazySP-Expand.
      After evaluating the same set of edges,
      the next edges to be evaluated by each algorithm
      can both be expressed as a surjective mapping onto
      a common set of unexpanded
      frontier vertices.
      }
   \label{fig:astar-equiv-mapping}
\end{figure}
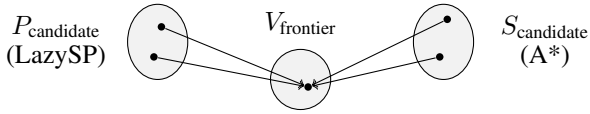

\subsubsection{Equivalence.}
In order to show edge-equivalence,
we consider the case where both algorithms
are beginning a new iteration
having so far evaluated the same set of edges.

LazySP-Expand has some set $P_{\ms{candidate}}$ of allowable
candidate paths minimizing $\mbox{len}(p,w_{\ms{lazy}})$;
the Expand selector will then identify a vertex on the chosen path
for expansion.

A* will iteratively select a set of vertices from OPEN to expand.
Because it is possible that a vertex is expanded multiple times
(and only the first expansion results in edge evaluations),
we group iterations of A* into \emph{sequences},
where each sequence $s$ consists of
(a) zero or more vertices from OPEN that have already been expanded,
followed by (b) one vertex from OPEN that is to be expanded
for the first time.

We show that both the set of allowable candidate paths $P_{\ms{candidate}}$
available to LazySP-Expand
and the set of allowable candidate vertex sequences $S_{\ms{candidate}}$
available to A*
map surjectively to the same set of unexpanded frontier vertices $V_{\ms{frontier}}$
as illustrated in Figure~\ref{fig:astar-equiv-mapping}.
This is established by way of
Theorems \ref{thm:astar-equiv-from-lazy}
and \ref{thm:astar-equiv-to-lazy} below.

\begin{theorem}
If LazySP-Expand and A* have evaluated the same set of edges,
then for any candidate path $p_{\ms{candidate}}$ chosen by LazySP
yielding frontier vertex $v_{\ms{frontier}}$,
there exists an allowable A* sequence $s_{\ms{candidate}}$
which also yields $v_{\ms{frontier}}$.
\label{thm:astar-equiv-from-lazy}
\end{theorem}

\begin{theorem}
If LazySP-Expand and A* have evaluated the same set of edges,
then for any candidate sequence $s_{\ms{candidate}}$ chosen by A*
yielding frontier vertex $v_{\ms{frontier}}$,
there exists an allowable LazySP path $p_{\ms{candidate}}$
which also yields $v_{\ms{frontier}}$.
\label{thm:astar-equiv-to-lazy}
\end{theorem}

\subsection{Equivalence to Lazy Weighted A*}

In a conventional vertex expansion algorithm,
determining a successor's cost is a function of both
the cost of the edge and the value of the heuristic.
If either of these components is expensive to evaluate,
an algorithm can defer its computation by maintaining the successor
on the frontier with an approximate cost until it is expanded.
The Fast Downward algorithm \cite{helmert2006fastdownward} is motivated
by expensive heuristic evaluations in planning,
whereas the Lazy Weighted A* (LWA*) algorithm \cite{cohen2014narms}
is motivated by expensive edge evaluations in robotics.

We show that the LazySP-Forward algorithm
is edge-equivalent to a variant of the Lazy Weighted A*
shortest-path algorithm.
For a given candidate path,
the Forward selector returns the first unevaluated edge.

\subsubsection{Variant of Lazy Weighted A*.}
We reproduce a variant of LWA* without a CLOSED list
in Algorithm~\ref{alg:lwastar}.
For the purposes of our analysis,
the reproduction differs from the original presentation,
and we detail those differences here.
With the exception of the lack of CLOSED,
the differences do not affect the behavior of the algorithm.

\begin{algorithm}[t]
\caption{Lazy Weighted A* (without CLOSED list)}
\label{alg:lwastar}
\begin{algorithmic}[1]
\Function {\textsc{LazyWeightedA*}}{$G, w, \hat{w}, h$}
\State $g[v_{\ms{start}}] \leftarrow 0$
   \Comment For uninitialized, $g[v] = \infty$
\State $Q_v \leftarrow \{ v_{\ms{start}} \}$
   \Comment Key: $g[v] + h(v)$%
   \label{line:lwastar-key-qvertices}
\State $Q_e \leftarrow \emptyset$
   \Comment Key: $g[v] + \hat{w}(v,v') + h(v')$%
   \label{line:lwastar-key-qedges}
\While {$\min(Q_v.{\mbox{TopKey}}, Q_e.{\mbox{TopKey}}) < g[v_{\ms{goal}}]$}
   \If {$Q_v.{\mbox{TopKey}} \leq Q_e.{\mbox{TopKey}}$}
      \State $v \leftarrow Q_v.{\mbox{Pop}}()$
      \For {$v' \in G.\mbox{GetSuccessors}(v)$}
         \State $Q_e.\mbox{Insert}((v,v'))$
      \EndFor
   \Else
      \State $(v,v') \leftarrow Q_e.{\mbox{Pop}}()$
      \If {$g[v'] \leq g[v] + \hat{w}(v,v')$}
         \label{line:lwastar-test}
         \State {\bf continue}
      \EndIf
      \State $g_{\ms{new}} \leftarrow g[v] + w(v,v')$
         \Comment evaluate
      \If {$g_{\ms{new}} < g[v']$}
         \State $g[v'] = g_{\ms{new}}$
         \State $Q_v.\mbox{Insert}(v')$
      \EndIf
   \EndIf
\EndWhile
\EndFunction
\end{algorithmic}
\end{algorithm}

The most obvious difference is that we present the original OPEN list
as separate vertex ($Q_v$) and edge ($Q_e$) priority queues,
with sorting keys shown on lines \ref{line:lwastar-key-qvertices}
and \ref{line:lwastar-key-qedges}.
A vertex $v$ in the original OPEN with $trueCost(v) = true$
corresponds to a vertex $v$ in $Q_v$,
whereas a vertex $v'$ in the original OPEN
with $trueCost(v') = false$ (and parent $v$)
corresponds to an edge $(v,v')$ in $Q_e$.
Use of the edge queue obviates the need for
duplicate vertices on OPEN with different parents
and the $conf(v)$ test for identifying such duplicates.
This presentation also highlights the similarity between LWA*
and the inner loop of the Batch Informed Trees (BIT*) algorithm
\cite{gammell2015bitstar}.

The second difference is that the edge usefulness test
(line 12 of the original algorithm)
has been moved from before inserting into OPEN
to after being popped from OPEN,
but before being evaluated
(line~\ref{line:lwastar-test} of Algorithm~\ref{alg:lwastar}).
This change is partially in compensation for removing the CLOSED
list.
This adjustment
does not affect the edges evaluated.

We make use of an invariant that is maintained during the
progression of Lazy Weighted A*.
\begin{invariant}
For all vertex pairs $v$ and $v'$,
with $v'$ a successor of $v$,
if $g[v] + \max(w(v,v'), \hat{w}(v,v')) < g[v']$,
then either vertex $v$ is on $Q_{v}$
or edge $(v,v')$ is on $Q_e$.%
\label{inv:lwastar}%
\end{invariant}
We will use $h(v) = h_{\ms{lazy}}(v)$ from (\ref{eqn:h_lazy})
and $\hat{w} = w_{\ms{lazy}}$.
Note that the use of these dynamic heuristics requires that the
$Q_v$ and $Q_e$ be resorted after every edge is evaluated.

\subsubsection{Equivalence.}
The equivalence follows similarly to that for A* above.
Given the same set of edges evaluated,
the set of allowable next evaluations is identical for each
algorithm.

\begin{theorem}
If LazySP-Forward and LWA* have evaluated the same set of edges,
then for any allowable candidate path $p_{\ms{candidate}}$
chosen by LazySP yielding first unevaluated edge $e_{ab}$,
there exists an allowable LWA* sequence $s_{\ms{candidate}}$
which also yields $e_{ab}$.
\label{thm:lwastar-equiv-from-lazy}
\end{theorem}

\begin{theorem}
If LazySP-Forward and LWA* have evaluated the same set of edges,
then for any allowable sequence of vertices and edges $s_{\ms{candidate}}$
considered by LWA* yielding evaluated edge $e_{ab}$,
there exists an allowable LazySP candidate path $p_{\ms{candidate}}$
which also yields $e_{ab}$.
\label{thm:lwastar-equiv-to-lazy}
\end{theorem}

\subsection{Relation to Bidirectional Heuristic Search}

LazySP-Alternate chooses unevaluated edges from either
the beginning or the end of the candidate path at each iteration.
We conjecture that an alternating version of the Expand selector
is edge-equivalent to the
Bidirectional Heuristic Front-to-Front Algorithm
\cite{sint1977bhffa}
for appropriate lazy vertex pair heuristic,
and that LazySP-Alternate is edge-equivalent
to a bidirectional LWA*.

\section{Novel Edge Selectors}

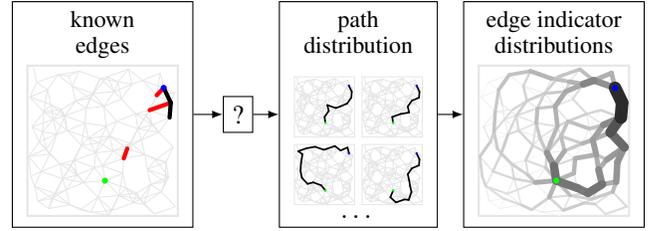
\begin{figure}[t]
   \centering
   \begin{tikzpicture}
      \tikzset{>=latex}
      
      \node[draw,minimum width=2.4cm,minimum height=3.0cm] (startbox) at (-3.0,0) {};
      \node[inner sep=0pt] at (-3.0,-0.35) {\includegraphics[scale=2.0]{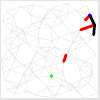}};
      \node[align=center,font=\small,below] at (startbox.north) {known\\edges};
      
      \node[draw] (quesbox) at (-1.2,0) {?};
      
      \node[draw,minimum width=2.1cm,minimum height=3.0cm] (pathsbox) at (0.4,0) {};
      \node[inner sep=0pt] at (-0.05, 0.1) {\includegraphics[scale=0.8]{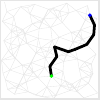}};
      \node[inner sep=0pt] at ( 0.85, 0.1) {\includegraphics[scale=0.8]{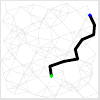}};
      \node[inner sep=0pt] at (-0.05,-0.8) {\includegraphics[scale=0.8]{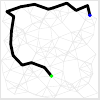}};
      \node[inner sep=0pt] at ( 0.85,-0.8) {\includegraphics[scale=0.8]{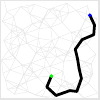}};
      \node[align=center,font=\small,below] at (pathsbox.north) {path\\distribution};
      \node[align=center,font=\normalsize,above] at (pathsbox.south) {$\dots$};
      
      \node[draw,minimum width=2.4cm,minimum height=3.0cm] (goalbox) at (3.0,0) {};
      \node[inner sep=0pt] at (3.0,-0.35) {\includegraphics[scale=2.0]{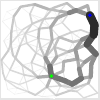}};
      \node[align=center,font=\small,below] at (goalbox.north) {edge indicator\\distributions};
      
      \draw[->] (startbox) -- (quesbox);
      \draw[->] (quesbox) -- (pathsbox);
      \draw[->] (pathsbox) -- (goalbox);
      
   \end{tikzpicture}
   \caption{Illustration of maximum edge probability selectors.
      A distribution over paths
      (usually conditioned on the known edge evaluations)
      induces on each edge $e$ a Bernoulli distribution
      with parameter $p(e)$
      giving the probability that it belongs to the path.
      The selector chooses the edge with the largest such probability.}
   \label{fig:maxprob-selectors-overview}
\end{figure}

\begin{algorithm}[t]
   \caption{Maximum Edge Probability Selector\\
      \emph{(for WeightSamp and Partition path distributions)}}
   \begin{algorithmic}[1]
   \Function {\textsc{SelectMaxEdgeProb}}{$G, p_{\ms{candidate}}, \mathcal{D}_p$}
   \State $p(e) \leftarrow \Pr( \, e \in P \, )
      \mbox{ for } P \sim \mathcal{D}_p$
   \State $e_{\ms{max}} \leftarrow$ unevaluated $e \in p_{\ms{candidate}}$
      maximizing $p(e)$
   \State \Return $\{ e_{\ms{max}} \}$
   \EndFunction
   \end{algorithmic}
   \label{alg:selectmaxscore}
\end{algorithm}

Because we are conducting a search over paths,
we are free to implement selectors which are not constrained to
evaluate edges in any particular order
(i.e. to maintain evaluated trees rooted at the start and goal
vertices).
In this section,
we describe a novel class of edge selectors which is designed
to reduce the total number of edges evaluated during the course
of the LazySP algorithm.
These selectors operate by maintaining a distribution over the
set of potential paths $P$ at each iteration of the algorithm
(see Figure~\ref{fig:maxprob-selectors-overview}).
Such a path distribution induces a Bernoulli distribution for each
edge $e$ which indicates its probability $p(e)$ to lie on
the potential path;
at each iteration,
the selectors then choose the unevaluated edge that maximizes
this edge indicator probability (Algorithm~\ref{alg:selectmaxscore}).
The two selectors described in this section differ
with respect to how they maintain this distribution over potential paths.

\subsection{Weight Function Sampling Selector}

The first selector, WeightSamp,
is motivated by the intuition that it is preferable to evaluate edges
that are most likely to lie on the true shortest path.
Therefore,
it computes its path distribution $\mathcal{D}_p$
by performing shortest path queries
on sampled edge weight functions drawn from a distribution
$\mathcal{D}_w$.
This edge weight distribution is conditioned on the the known weights
of all previously evaluated edges $E_{\ms{eval}}$:
\begin{equation}
   \mathcal{D}_p : \mbox{SP}(w)
   \mbox{ for } w \sim \mathcal{D}_w(E_{\ms{eval}})
   \label{eqn:weightsamp}.
\end{equation}

\begin{figure}[t]
   \centering
   \begin{tikzpicture}
      \tikzset{>=latex}
      
      \node[draw,minimum width=1.8cm,minimum height=2.6cm] (startbox) at (-4.4,0) {};
      \node[inner sep=0pt] at (-4.4,-0.35) {\includegraphics[scale=1.5]{build/figs/lazysp-fig-dists/fig-sofar}};
      \node[align=center,font=\small,below] at (startbox.north) {known\\edges};
      
      \node[draw,minimum width=1.8cm,minimum height=6cm] (abox) at (-2.2,0) {};
      \node[inner sep=0pt] at (-2.2, 1.3) {\includegraphics[scale=1.5]{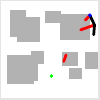}};
      \node[inner sep=0pt] at (-2.2,-0.3) {\includegraphics[scale=1.5]{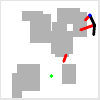}};
      \node[inner sep=0pt] at (-2.2,-1.9) {\includegraphics[scale=1.5]{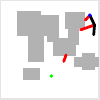}};
      \node[align=center,font=\small,below] at (abox.north) {obstacle\\distribution};
      \node[align=center,font=\normalsize,above] at (abox.south) {$\dots$};
      
      \node[draw,minimum width=1.8cm,minimum height=6cm] (bbox) at (0,0) {};
      \node[inner sep=0pt] at (0, 1.3) {\includegraphics[scale=1.5]{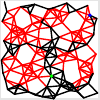}};
      \node[inner sep=0pt] at (0,-0.3) {\includegraphics[scale=1.5]{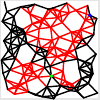}};
      \node[inner sep=0pt] at (0,-1.9) {\includegraphics[scale=1.5]{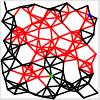}};
      \node[align=center,font=\small,below] at (bbox.north) {weight fn\\distribution};
      \node[align=center,font=\normalsize,above] at (bbox.south) {$\dots$};
      
      \node[draw,minimum width=1.8cm,minimum height=6cm] (cbox) at (2.2,0) {};
      \node[inner sep=0pt] at (2.2, 1.3) {\includegraphics[scale=1.5]{build/figs/lazysp-fig-dists/fig-path-00}};
      \node[inner sep=0pt] at (2.2,-0.3) {\includegraphics[scale=1.5]{build/figs/lazysp-fig-dists/fig-path-01}};
      \node[inner sep=0pt] at (2.2,-1.9) {\includegraphics[scale=1.5]{build/figs/lazysp-fig-dists/fig-path-02}};
      \node[align=center,font=\small,below] at (cbox.north) {path\\distribution};
      \node[align=center,font=\normalsize,above] at (cbox.south) {$\dots$};
      
      \draw[->] (startbox) -- (abox);
      \draw[->] (abox) -- (bbox);
      \draw[->] (bbox) -- (cbox);
   \end{tikzpicture}
   \caption{The WeightSamp selector uses the path distribution induced by
      solving the shortest path problem on a distribution over possible
      edge weight functions $\mathcal{D}_w$.
      In this example, samples from $\mathcal{D}_w$ are computed by
      drawing samples from $\mathcal{D}_O$,
      the distribution of obstacles that are consistent with
      the known edge evaluations.}
   \label{fig:weightsamp}
\end{figure}

For example,
the distribution $\mathcal{D}_w$ might consist of
the edge weights induced by a model of the distribution of
environment obstacles
(Figure~\ref{fig:weightsamp}).
Since this obstacle distribution is conditioned on the results
of known edge evaluations,
we consider the subset of worlds which are consistent
with the edges we have evaluated so far.
However,
depending on the fidelity of this model,
solving the corresponding shortest path problem for a given
sampled obstacle arrangement might require as much computation as
solving the original problem,
since it requires computing the resulting edge weights.
In practice,
we can approximate $\mathcal{D}_w$
by assuming that each edge is independently distributed.

\subsection{Partition Function Selector}

While the WeightSamp selector captures the intuition that it is
preferable to focus edge evaluations in areas that are useful for
many potential paths,
the computational cost required to calculate it at each iteration
may render it intractable.
One candidate path distribution that is more efficient to compute
has a density function which follows an exponential form:
\begin{equation}
   \mathcal{D}_p : f_P(p) \propto
   \exp( - \beta \, \mbox{len}(p, w_{\ms{lazy}}) ).
\end{equation}
In other words,
we consider all potential paths $P$
between the start and goal vertices,
with shorter paths assigned more probability than longer ones
(with positive parameter $\beta$).
We call this the Partition selector
because this distribution is closely related to calculating
partition functions from statistical mechanics.
The corresponding partition function is:
\begin{equation}
   Z(P) = \sum_{p \in P}
      \exp( - \beta \, \mbox{len}(p, w_{\ms{lazy}}) ).
   \label{eqn:partitionfn}
\end{equation}
Note that the edge indicator probability
required in Algorithm~\ref{alg:selectmaxscore}
can then be written:
\begin{equation}
   p(e) = 1 - \frac{Z(P \setminus e)}{Z(P)}.
   \label{eqn:edge-ind-prob}
\end{equation}
Here, $P \setminus e$ denotes paths in $P$ that do not
contain edge $e$.

\begin{figure}[t!]
   \centering
   \begin{subfigure}[b]{8.5cm}
      \centering
      \includegraphics{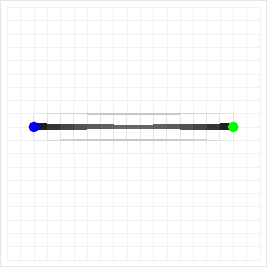}
      \includegraphics{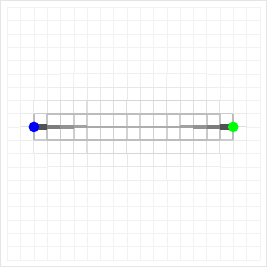}
      \includegraphics{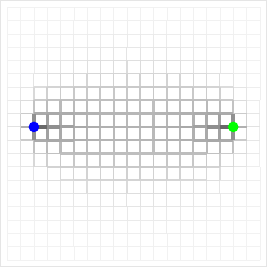}
      \vspace{-0.05in}
      \caption{Initial $p(e)$ scores on a constant-weight
         grid with $\beta$: 50, 33, 28}
      \vspace{0.05in}
      \label{subfig:partition-empty}
   \end{subfigure}
   \begin{subfigure}[b]{8.5cm}
      \centering
      \includegraphics{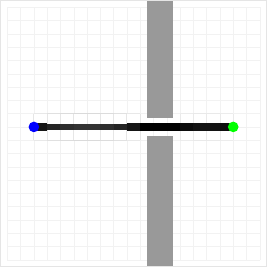}
      \includegraphics{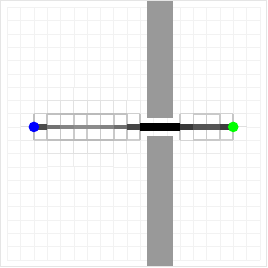}
      \includegraphics{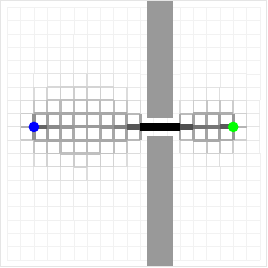}
      \vspace{-0.05in}
      \caption{Initial $p(e)$ scores with $\infty$-weight
         obstacles with $\beta$: 50, 33, 28}
      \vspace{0.05in}
      \label{subfig:partition-passage}
   \end{subfigure}
   \begin{subfigure}[b]{4cm}
      \centering
      \includegraphics{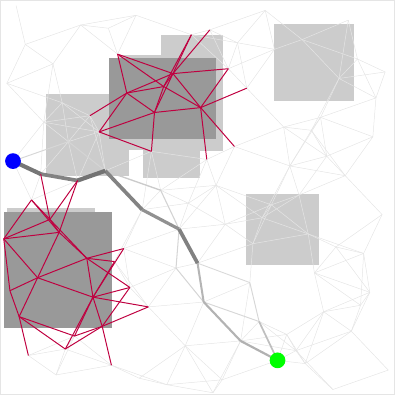}
      \caption{Initial $p(e)$ scores}
      \label{subfig:partition-example-initial}
   \end{subfigure}
   \;
   \begin{subfigure}[b]{4cm}
      \centering
      \includegraphics{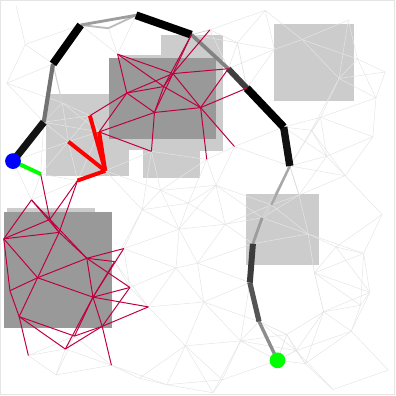}
      \caption{Scores after five evaluations}
      \label{subfig:partition-example-after5}
   \end{subfigure}
   \caption{Examples of the Partition selector's
      $p(e)$ edge score function.
      (\subref{subfig:partition-empty}) With no known obstacles,
      a high $\beta$ assigns near-unity score to only edges on the
      shortest path;
      as $\beta$ decreases and more paths are considered,
      edges immediately adjacent to the roots score highest.
      (\subref{subfig:partition-passage}) Since all paths must pass
      through the narrow passage,
      edges within score highly.
      (\subref{subfig:partition-example-initial})
      For a problem with two a-priori known obstacles (dark gray),
      the score first prioritizes evaluations between the two.
      (\subref{subfig:partition-example-after5})
      Upon finding these edges are blocked,
      the next edges that are prioritized lie along the top of the world.}
   \label{ref:example-scores}
\end{figure}

\begin{figure*}[t!]
\centering
\includegraphics[width=3.49cm]{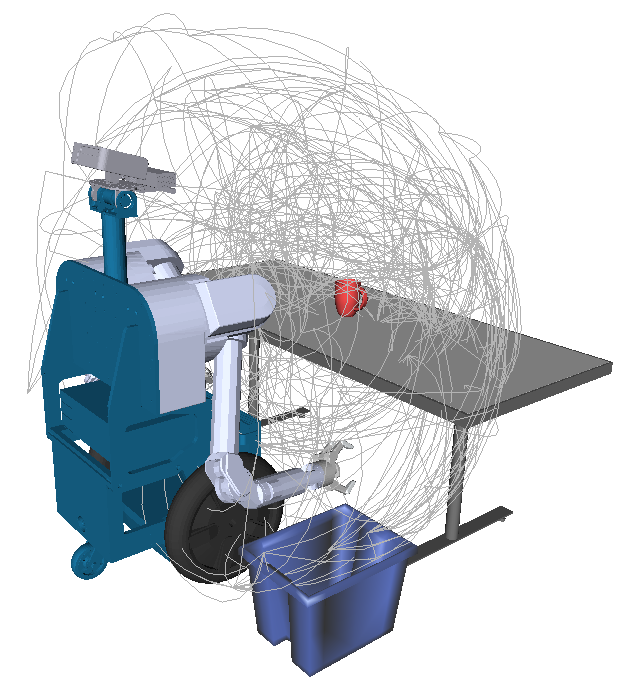}
\includegraphics[width=3.49cm]{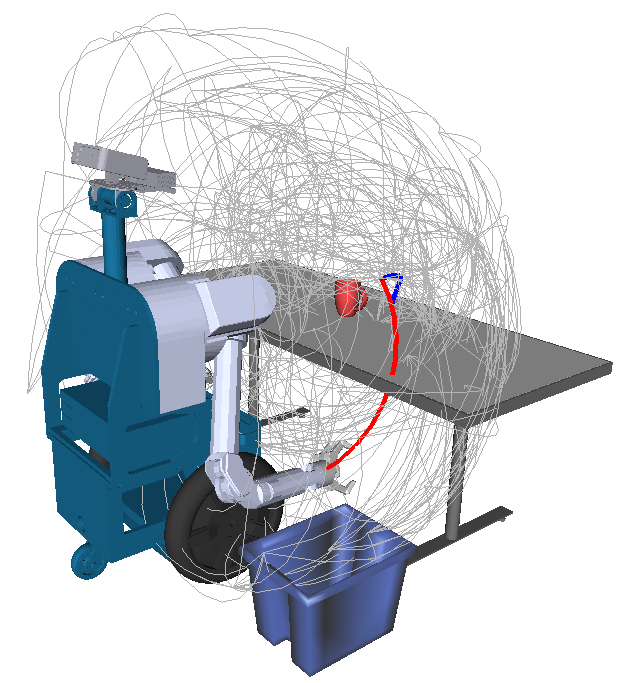}
\includegraphics[width=3.49cm]{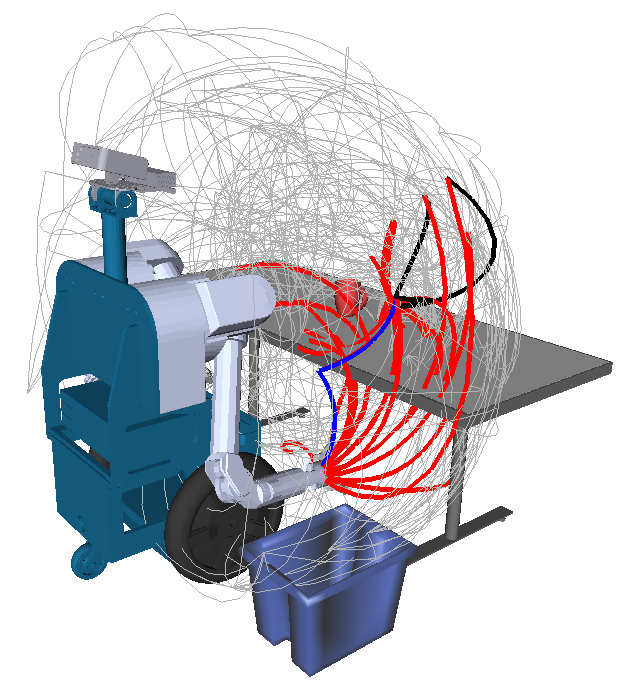}
\includegraphics[width=3.49cm]{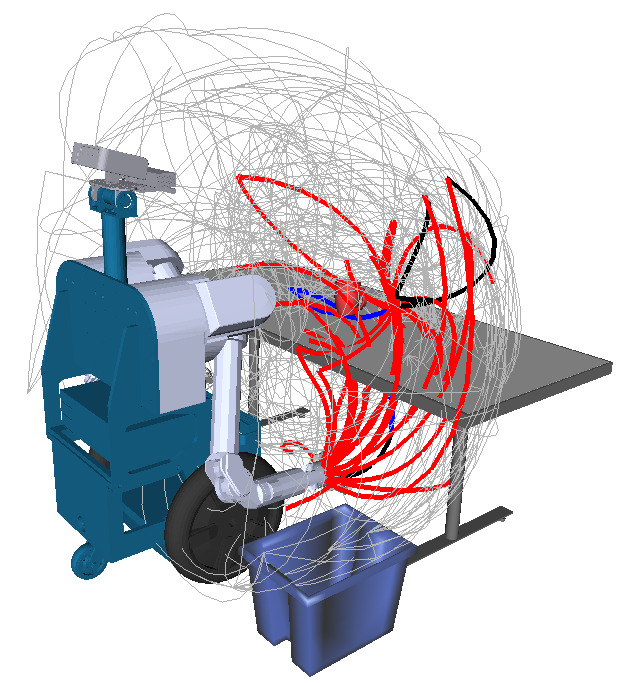}
\includegraphics[width=3.49cm]{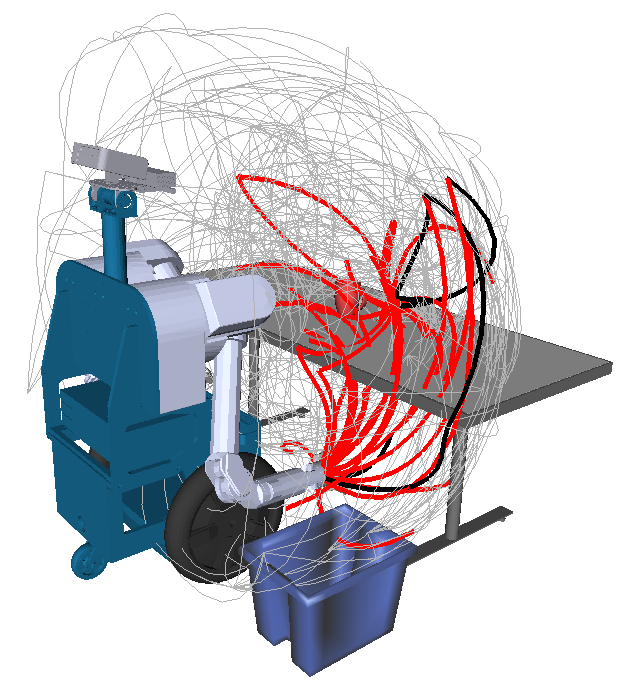}
\caption{Visualization of the first of three articulated motion
   planning problems in which the HERB robot must move its right arm
   from the start configuration (pictured)
   to any of seven grasp configurations for a mug.
   Shown is the progression of the Alternate selector on one of the
   randomly generated roadmaps;
   approximately 2\% of the 7D roadmap is shown in gray by projecting
   onto the space of end-effector positions.}
\label{fig:herbbin0}
\end{figure*}

It may appear advantageous to restrict $P$ to only
\emph{simple} paths,
since all optimal paths are simple.
Unfortunately,
an algorithm for computing (\ref{eqn:partitionfn}) efficiently is not
currently known in this case.
However,
in the case that $P$ consists of all paths,
there exists an efficient incremental calculation of
(\ref{eqn:partitionfn}) via a recursive formulation
which we detail here.

We use the notation $Z_{xy} = Z(P_{xy})$,
with $P_{xy}$ the set of paths from $x$ to $y$.
Suppose the values $Z_{xy}$ are known between
all pairs of vertices $x, y$ for a graph $G$.
(For a graph with no edges,
$Z_{xy}$ is 1 if $x = y$ and 0 otherwise.)
Consider a modified graph $G'$ with one additional edge $e_{ab}$
with weight $w_{ab}$.
All additional paths use the new edge $e_{ab}$ a non-zero
number of times;
the value $Z'_{xy}$ can be shown to be
\begin{equation}
   Z'_{xy} = Z_{xy} + \frac{Z_{xa} Z_{by}}{\exp(\beta w_{ab}) - Z_{ba}}
   \mbox{ if }
   \exp(\beta w_{ab}) > Z_{ba}.
\end{equation}
This form is derived from simplifying the induced geometric series;
note that if $\exp(\beta w_{ab})  \leq Z_{ba}$,
the value $Z'_{xy}$ is infinite.
One can also derive the inverse:
given values $Z'$,
calculate the values $Z$ if an edge were removed.

This incremental formulation of (\ref{eqn:partitionfn})
allows for the corresponding score $p(e)$ for edges
to be updated efficiently during each iteration of LazySP as
the $w_{\ms{lazy}}$ value for edges chosen for evaluation are updated.
In fact,
if the values $Z$ are stored in a square matrix,
the update for all pairs after an edge weight change consists of a single
vector outer product.

\section{Experiments}

We compared the seven edge selectors on three classes of shortest path
problems.
The average number of edges evaluated by each,
as well as timing results from our implementations,
are shown in Figure~\ref{fig:results}.
In each case,
the estimate was chosen so that $w_{\ms{est}} \leq w$,
so that all runs produced optimal paths.
The experimental results serve primarily to illustrate that
the A* and LWA* algorithms
(i.e. Expand and Forward)
are not optimally edge-efficient,
but they also expose differences in behavior and prompt
future research directions.
All experiments were conducted using an open-source
implementation.\footnote{%
https://github.com/personalrobotics/lemur}
Motion planning results were implemented using
OMPL \cite{sucan2012ompl}.

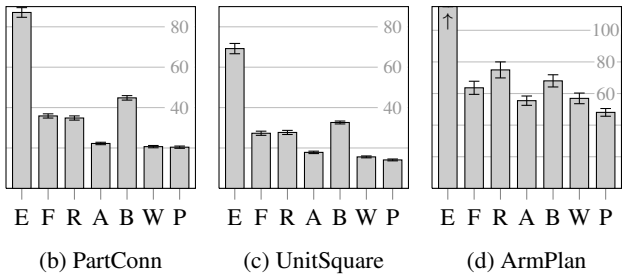
\begin{figure}[t!]
   \centering
   \begin{subfigure}[b]{\columnwidth}%
      \centering
      {\small%
      \setlength{\tabcolsep}{0.06cm}%
      \begin{tabular}{lrrrrrrr}
         \toprule
            & E\;\;\;\;
            & F\;\;\;\; & R\;\;\;\; & A\;\;\;\;
            & B\;\;\;\; & W\;\;\;\; & P\ddag\;\; \\
         \midrule
         \addlinespace[0.3em]
         PartConn &  87.10 & 35.86 & 34.84 & 22.23 & 44.81 & \textbf{20.66} & \textbf{20.39} \\
         \;\;\emph{online\dag (ms)} & \bf\emph{1.22} & \emph{1.96} & \emph{1.86} & \bf\emph{1.20} & \emph{2.41} & \emph{4807.19} & \emph{3.32} \\
         \;\;\;\;\emph{sel (ms)} & \emph{0.02} & \emph{0.01} & \emph{0.01} & \emph{0.01} & \emph{0.03} & \emph{4805.64} & \emph{2.07} \\
         \addlinespace[0.3em]
         UnitSquare &  69.21 & 27.29 & 27.69 & 17.82 & 32.62 & 15.58 & \textbf{14.08} \\
         \;\;\emph{online\dag (ms)} & \bf\emph{0.91} & \emph{1.47} & \emph{1.49} & \bf\emph{0.94} & \emph{1.71} & \emph{3864.95} & \emph{1.72} \\
         \;\;\;\;\emph{sel (ms)} & \emph{0.01} & \emph{0.01} & \emph{0.01} & \emph{0.01} & \emph{0.02} & \emph{3863.49} & \emph{0.87} \\
         \addlinespace[0.3em]
         ArmPlan(avg) & 949.05 & 63.62 & 74.94 & 55.48 & 68.01 & 56.93 & \textbf{48.07} \\
         \;\;\emph{online (s)} & \emph{269.82} & \bf\emph{5.90} & \emph{8.22} & \bf\emph{5.96} & \emph{7.34} & \emph{3402.21} & \bf\emph{5.80} \\
         \;\;\;\;\emph{sel (s)} & \emph{0.00} & \emph{0.00} & \emph{0.00} & \emph{0.00} & \emph{0.00} & \emph{3392.76} & \emph{1.54} \\
         \;\;\;\;\emph{eval (s)} & \emph{269.78} & \emph{5.87} & \emph{8.20} & \emph{5.94} & \emph{7.31} & \emph{9.39} & \emph{4.21} \\
         \addlinespace[0.3em]
         ArmPlan1 &  344.74 & \textbf{49.72} & 95.58 & 59.44 & 58.90 & 73.72 & \textbf{50.66} \\
         \;\;\emph{online (s)} & \emph{109.09} & \bf\emph{4.81} & \emph{14.81} & \emph{7.03} & \emph{7.91} & \emph{3375.35} & \emph{7.25} \\
         \;\;\;\;\emph{sel (s)} & \emph{0.00} & \emph{0.00} & \emph{0.00} & \emph{0.00} & \emph{0.00} & \emph{3358.82} & \emph{1.61} \\
         \;\;\;\;\emph{eval (s)} & \emph{109.07} & \emph{4.78} & \emph{14.77} & \emph{7.01} & \emph{7.88} & \emph{16.47} & \emph{5.59} \\
         \addlinespace[0.3em]
         ArmPlan2 &  657.02 & \textbf{62.24} & 98.54 & 69.96 & 75.88 & \textbf{66.24} & \textbf{62.16} \\
         \;\;\emph{online (s)} & \emph{166.19} & \bf\emph{3.27} & \emph{7.36} & \emph{5.95} & \emph{5.63} & \emph{4758.04} & \emph{5.99} \\
         \;\;\;\;\emph{sel (s)} & \emph{0.00} & \emph{0.00} & \emph{0.00} & \emph{0.00} & \emph{0.00} & \emph{4750.16} & \emph{2.03} \\
         \;\;\;\;\emph{eval (s)} & \emph{166.17} & \emph{3.26} & \emph{7.34} & \emph{5.93} & \emph{5.61} & \emph{7.82} & \emph{3.91} \\
         \addlinespace[0.3em]
         ArmPlan3 & 1845.38 & 78.90 & \textbf{30.70} & 37.04 & 69.26 & \textbf{30.82} & \textbf{31.38} \\
         \;\;\emph{online (s)} & \emph{534.16} & \emph{9.61} & \bf\emph{2.50} & \emph{4.91} & \emph{8.47} & \emph{2073.23} & \emph{4.17} \\
         \;\;\;\;\emph{sel (s)} & \emph{0.00} & \emph{0.00} & \emph{0.00} & \emph{0.00} & \emph{0.00} & \emph{2069.29} & \emph{0.98} \\
         \;\;\;\;\emph{eval (s)} & \emph{534.10} & \emph{9.58} & \emph{2.48} & \emph{4.89} & \emph{8.44} & \emph{3.90} & \emph{3.15} \\
         \addlinespace[0.15em]
         \bottomrule
      \end{tabular}%
      }%
      \caption{Average number of edges evaluated for each problem class
         and selector.
         The minimum selector,
         along with any selector within one unit of its standard error,
         is shown in bold.
         The ArmPlan class is split into its three constituent problems.
         Online timing results are also shown,
         including the components from the invoking the selector
         and evaluating edges.
         \dag PartConn and UnitSquare involve trivial edge evaluation
         time.
         \ddag Timing for the Partition selector does not include
         pre-computation time.}
      \label{subfig:table-results}
   \end{subfigure}
   
   \vspace{0.1in}
   
   \begin{subfigure}[b]{2.75cm}
      \centering
      \begin{tikzpicture}
      \begin{axis}[
         width=4.1cm,
         height=4.0cm,
         ybar,
         bar width=7,
         ymin=0,ymax=90,
         ytick pos=bottom,
         symbolic x coords={E, F, R, A, B, W, P},
         xtick=data,
         xtick pos=left,
         ymajorgrids,
         ymajorticks=false,
         ticklabel style={font=\small}
         ]
      \node[circle,fill=white,inner sep=1pt,text=black!40] at (axis cs:P,40) {\scriptsize 40};
      \node[circle,fill=white,inner sep=1pt,text=black!40] at (axis cs:P,60) {\scriptsize 60};
      \node[circle,fill=white,inner sep=1pt,text=black!40] at (axis cs:P,80) {\scriptsize 80};
      \addplot[color=black,fill=black!20,error bars/.cd,y dir=both,y explicit] coordinates {
         (E, 87.10) +- (2.39,2.39)
         (F, 35.86) +- (1.04,1.04)
         (R, 34.84) +- (1.04,1.04)
         (A, 22.23) +- (0.60,0.60)
         (B, 44.81) +- (1.11,1.11)
         (W, 20.66) +- (0.57,0.57)
         (P, 20.39) +- (0.56,0.56)
      };
      \end{axis}
      \end{tikzpicture}
      \caption{PartConn}
   \end{subfigure}
   \begin{subfigure}[b]{2.75cm}
      \centering
      \begin{tikzpicture}
      \begin{axis}[
         width=4.1cm,
         height=4.0cm,
         ybar,
         bar width=7,
         ymin=0,ymax=90,
         ytick pos=bottom,
         symbolic x coords={E, F, R, A, B, W, P},
         xtick=data,
         xtick pos=left,
         ymajorgrids,
         ymajorticks=false,
         ticklabel style={font=\small}
         ]
      \node[circle,fill=white,inner sep=1pt,text=black!40] at (axis cs:P,40) {\scriptsize 40};
      \node[circle,fill=white,inner sep=1pt,text=black!40] at (axis cs:P,60) {\scriptsize 60};
      \node[circle,fill=white,inner sep=1pt,text=black!40] at (axis cs:P,80) {\scriptsize 80};
      \addplot[color=black,fill=black!20,error bars/.cd,y dir=both,y explicit] coordinates {
         (E, 69.21) +- (2.55,2.55)
         (F, 27.29) +- (1.03,1.03)
         (R, 27.69) +- (1.02,1.02)
         (A, 17.82) +- (0.60,0.60)
         (B, 32.62) +- (0.72,0.72)
         (W, 15.58) +- (0.47,0.47)
         (P, 14.08) +- (0.46,0.46)
      };
      \end{axis}
      \end{tikzpicture}
      \caption{UnitSquare}
   \end{subfigure}
   \begin{subfigure}[b]{2.75cm}
      \centering
      \begin{tikzpicture}
      \begin{axis}[
         width=4.1cm,
         height=4.0cm,
         ybar,
         bar width=7,
         ymin=0,ymax=115,
         max space between ticks=10,
         ytick pos=bottom,
         symbolic x coords={E, F, R, A, B, W, P},
         xtick=data,
         xtick pos=left,
         ymajorgrids,
         ymajorticks=false,
         ticklabel style={font=\small}
         ]
      \node[circle,fill=white,inner sep=0pt,text=black!40] at (axis cs:P,60) {\scriptsize 60};
      \node[circle,fill=white,inner sep=0pt,text=black!40] at (axis cs:P,80) {\scriptsize 80};
      \node[circle,fill=white,inner sep=0pt,text=black!40] at (axis cs:P,100) {\scriptsize 100};
      \addplot[color=black,fill=black!20,error bars/.cd,y dir=both,y explicit] coordinates {
         (E, 115) +- (0,0) 
         (F, 63.62) +- (4.15,4.15)
         (R, 74.94) +- (5.07,5.07)
         (A, 55.48) +- (2.95,2.95)
         (B, 68.01) +- (3.86,3.86)
         (W, 56.93) +- (3.37,3.37)
         (P, 48.07) +- (2.44,2.44)
      };
      \node[align=center,anchor=north,inner sep=0pt] at (axis cs:E,111) {\scriptsize $\uparrow$};
      \end{axis}
      \end{tikzpicture}
      \caption{ArmPlan}
   \end{subfigure}
   \caption{
      Experimental results for the three problem classes
      across each of the seven selectors,
      E:Expand, F:Forward, R:Reverse,
      A:Alternate, B:Bisection,
      W:WeightSamp, and P:Partition.
      In addition to the summary table (a),
      the plots (b-d) show summary statistics for
      each problem class.
      The means and standard errors in (b-c) are across the
      1000 and 900 problem instances, respectively.
      The means and standard errors in (d) are for
      the average across the three constituent problems
      for each of the 50 sampled roadmaps.}
   \label{fig:results}
\end{figure}

\subsubsection{Random partially-connected graphs.}
We tested on a set of 1000 randomly-generated undirected graphs
with $|V|=100$,
with each pair of vertices sharing an edge with probability 0.05.
Edges have an independent 0.5 probability of having infinite weight,
else the weight is uniformly distributed on $[1,2]$;
the estimated weight was unity for all edges.
For the WeightSamp selector,
we drew 1000 $w$ samples at each iteration
from the above edge weight distribution.
For the Partition selector, we used $\beta = 2$.

\subsubsection{Roadmap graphs on the unit square.}
We considered roadmap graphs formed via the first 100 points
of the $(2,3)$-Halton sequence on the unit square
with a connection radius of 0.15,
with 30 pairs of start and goal vertices chosen randomly.
The edge weight function was derived from 30 sampled obstacle fields
consisting of 10 randomly placed
axis-aligned boxes with dimensions uniform on $[0.1,0.3]$,
with each edge having infinite weight on collision,
and weight equal to its Euclidean length otherwise.
One of the resulting 900 example problems is shown in
Figure~\ref{fig:snapshots}.
For the WeightSamp selector,
we drew 1000 $w$ samples
with a na\"{\i}ve edge weight distribution
with each having an independent 0.1 collision probability.
For the Partition selector, we used $\beta = 21$.

\subsubsection{Roadmap graphs for robot arm motion planning.}
We considered roadmap graphs in the configuration space
corresponding to the 7-DOF right arm of the HERB home robot
\cite{srinivasa2012herb20} across
three motion planning problems inspired by a table clearing scenario
(see Figure~\ref{fig:herbbin0}).
The problems consisted of first moving from the robot's
home configuration to one of 7 feasible grasp configurations for a mug
(pictured),
second transferring the mug to one of 72 feasible configurations with
the mug above the blue bin,
and third returning to the home configuration.
Each problem was solved independently.
This common scenario spans various numbers of starts/goals
and allows a comparison w.r.t. difficulty at different problem
stages as discussed later.

For each problem,
50 random graphs were constructed by applying a random offset to
the 7D Halton sequence with $N = 1000$,
with additional vertices for each problem start and goal configuration.
We used an edge connection radius of 3 radians,
resulting $|E|$ ranging from 23404 to 28109.
Each edge took infinite weight on collision,
and weight equal to its Euclidean length otherwise.
For the WeightSamp selector,
we drew 1000 $w$ samples
with a na\"{\i}ve edge weight distribution in which
each edge had an independent 0.1 probability of collision.
For the Partition selector, we used $\beta = 3$.

\section{Discussion}
\label{sec:discussion}

The first observation that is evident from the experimental results
is that lazy evaluation
-- whether using Forward (LWA*) or one of the other selectors --
grossly outperforms Expand (A*).
The relative penalty that Expand incurs by evaluating all edges from
each expanded vertex is a function of the graph's branching factor.

Since the Forward and Reverse selectors are simply mirrors of each
other,
they exhibit similar performance
averaged across the PartConn and UnitSquare problem classes,
which are symmetric.
However,
this need not the case for a particular instance.
For example,
the start of ArmPlan1 and the goal of ArmPlan3 consist
of the arm's single home configuration in a relatively confined space.
As shown in the table in Figure~\ref{subfig:table-results},
it appears that the better selector on these problems attempts
to solve the more constrained side of the problem first.
While it may be difficult to determine a priori which part of the
problem will be the most constrained,
the simple Alternate selector's respectable performance
suggests that it may be a reasonable compromise.

The per-path plots at the bottom of Figure~\ref{fig:snapshots}
allow us to characterize the selectors' behavior.
For example,
Alternate often evaluates several edges on each path before finding
an obstacle.
Its early evaluations also tend to be useful later,
and it terminates after considering 10 paths on the illustrated problem.
In contrast, Bisection exhibits a fail-fast strategy,
quickly invalidating most paths after a single evaluation,
but needing 16 such paths (with very little reuse)
before it terminates.
In general, the Bisection selector did not outperform any of the
lazy selectors in terms of number of edges evaluated.
However,
it may be well suited to problem domains in which
evaluations that fail tend be less costly.

The novel selectors based on path distributions tend to minimize
edge evaluations on the problems we considered.
While the WeightSamp selector performs similarly to Partition on the
simpler problems,
it performs less well in the ArmPlan domain.
This may be because many more samples are needed to approximate
the requisite path distribution.

The path distribution selectors are motivated by focusing evaluation
effort in areas that are useful for many distinct candidate paths,
as illustrated in Figure~\ref{ref:example-scores}.
Note that in the absence of a priori knowledge,
the edges nearest to the start and goal tend to have the highest
$p(e)$ score,
since they are members of many potential paths.
Because it tends to focus evaluations in a similar way,
the Alternate selector may serve as a simple proxy for the
more complex selectors.

We note that
an optimal edge selector could be theoretically achieved by posing the
edge selection problem as a POMDP,
given a probabilistic model of the true costs.
While likely intractable in complex domains,
exploring this solution may yield useful approximations or insights.

\subsubsection{Timing results.}
Figure~\ref{subfig:table-results}
shows that the five simple selectors incur a negligible
proportion of the algorithm's runtime.
The WeightSamp and Partition selectors both require additional
time (especially the former) in order to reduce the time spent
evaluating edges.
This tradeoff depends intimately on the problem domain considered.
In the ArmPlan problem class,
the Partition selector was able to reduce average total online runtime
slightly despite an additional 1.54s of selector time.
Note that Partition requires an expensive computation of the graph's
initial $Z$-values,
which are independent of the true weights and start/goal vertices
(and can therefore be pre-computed, e.g. across all ArmPlan instances).
\ifx\aaaiversion\undefined
Full timing results are available in the appendix
(Figure~\ref{fig:table-timing-results}).
\else
Full timing results are available in \cite{dellin2016lazyspextended}.
\fi

\subsubsection{Optimizations.}
While we have focused on edge evaluations as the dominant source
of computational cost,
other considerations may also be important.
There are a number of optimizations that allow for efficient
implementation of LazySP.

The first relates to the repeated invocations of the inner
shortest path algorithm (line~\ref{line:lazy-outline-shortestpath}
of Algorithm~\ref{alg:lazy-outline}).
Because only a small number of edges change
weights between invocations,
an incremental search algorithm such as SSSP \cite{ramalingam1996}
or LPA* \cite{koenig2004lpastar} can be used to greatly improve
the speed of the inner searches.
Since the edge selector determines where on the graph edges are
evaluated,
the choices of the selector and the inner search algorithm
are related.
For example,
using the Forward selector with an incremental inner search
rooted at the goal
results in behavior similar to D* \cite{stentz1994dstar}
(albeit without the need to handle a moving start location)
since a large portion of the inner tree can be reused.

An optimization commonly applied to vertex searches called
\emph{immediate expansion} is also applicable to LazySP.
If an edge is evaluated with weight $w \leq w_{\ms{est}}$,
the inner search need not be run again before the next edge
is evaluated.

A third optimization is applicable to domains with infinite edge costs
(e.g. to represent infeasible transitions).
If the length of the path returned by the inner shortest path
algorithm is infinite,
LazySP can return this path immediately even if some of its
edges remain unevaluated
without affecting its (sub)optimality.
This reduces the number of edge evaluations needed in the case that
no feasible path exists.

\subsubsection{Other methods for expensive edge evaluations.}
An alternative to lazy evaluations is based on the observation that
when solved by vertex expansion algorithms,
such problems are characterized by slow vertex expansions.
To mitigate this,
approaches such as Parallel A* \cite{irani1986parallelastar}
and Parallel A* for Slow Expansions \cite{phillips2014pastarse}
aim to parallelize such expansions.
We believe that a similar approach can be applied to LazySP.

Another approach to finding short paths quickly is to relax the
optimization objective (\ref{eqn:objective}) itself.
While LazySP already supports a \emph{bounded-suboptimal} objective
via an inflated edge weight estimate
(Theorem~\ref{thm:lazy-optimality}),
it may also be possible to adapt the algorithm to address
\emph{bounded-cost} problems \cite{stern2011pts}.

\subsection{Acknowledgements}
We would like to thank Aaron Johnson and Michael Koval
for their comments on drafts of this work.
This work was (partially) funded by
the National Science Foundation IIS (\#1409003),
Toyota Motor Engineering \& Manufacturing (TEMA),
and the Office of Naval Research.

\bibliography{references}

\begin{thebibliography}{}

\bibitem[\protect\citeauthoryear{Bohlin and Kavraki}{2000}]{bohlin2000lazyprm}
Bohlin, R., and Kavraki, E.
\newblock 2000.
\newblock Path planning using {L}azy {PRM}.
\newblock In {\em {IEEE} International Conference on Robotics and Automation},
  volume~1,  521--528.

\bibitem[\protect\citeauthoryear{Cohen, Phillips, and
  Likhachev}{2014}]{cohen2014narms}
Cohen, B.; Phillips, M.; and Likhachev, M.
\newblock 2014.
\newblock Planning single-arm manipulations with n-arm robots.
\newblock In {\em Robotics: Science and Systems}.

\bibitem[\protect\citeauthoryear{Dijkstra}{1959}]{dijkstra1959anote}
Dijkstra, E.~W.
\newblock 1959.
\newblock A note on two problems in connexion with graphs.
\newblock {\em Numerishe Mathematik} 1(1):269--271.

\bibitem[\protect\citeauthoryear{Gammell, Srinivasa, and
  Barfoot}{2015}]{gammell2015bitstar}
Gammell, J.; Srinivasa, S.; and Barfoot, T.
\newblock 2015.
\newblock {B}atch {I}nformed {T}rees {(BIT*)}: Sampling-based optimal planning
  via the heuristically guided search of implicit random geometric graphs.
\newblock In {\em {IEEE} International Conference on Robotics and Automation},
  3067--3074.

\bibitem[\protect\citeauthoryear{Hart, Nilsson, and
  Raphael}{1968}]{hart1968astar}
Hart, P.; Nilsson, N.; and Raphael, B.
\newblock 1968.
\newblock A formal basis for the heuristic determination of minimum cost paths.
\newblock {\em {IEEE} Transactions on Systems Science and Cybernetics}
  4(2):100--107.

\bibitem[\protect\citeauthoryear{Helmert}{2006}]{helmert2006fastdownward}
Helmert, M.
\newblock 2006.
\newblock The fast downward planning system.
\newblock {\em Artificial Intelligence Research} 26:191--246.

\bibitem[\protect\citeauthoryear{Irani and Shih}{1986}]{irani1986parallelastar}
Irani, K.~B., and Shih, Y.
\newblock 1986.
\newblock Parallel {A}* and {AO}* algorithms: An optimality criterion and
  performance evaluation.
\newblock In {\em International Conference on Parallel Processing},  274--277.

\bibitem[\protect\citeauthoryear{Kavraki \bgroup et al\mbox.\egroup
  }{1996}]{kavrakietal1996prm}
Kavraki, L.; Svestka, P.; Latombe, J.-C.; and Overmars, M.
\newblock 1996.
\newblock Probabilistic roadmaps for path planning in high-dimensional
  configuration spaces.
\newblock {\em IEEE Transactions on Robotics and Automation} 12(4):566--580.

\bibitem[\protect\citeauthoryear{Koenig, Likhachev, and
  Furcy}{2004}]{koenig2004lpastar}
Koenig, S.; Likhachev, M.; and Furcy, D.
\newblock 2004.
\newblock Lifelong planning {A}*.
\newblock {\em Artificial Intelligence} 155(1--2):93--146.

\bibitem[\protect\citeauthoryear{Korf}{1985}]{korf1985idastar}
Korf, R.~E.
\newblock 1985.
\newblock Depth-first iterative-deepening: An optimal admissible tree search.
\newblock {\em Artificial Intelligence} 27:97--109.

\bibitem[\protect\citeauthoryear{Kuffner and
  LaValle}{2000}]{kuffner2000rrtconnect}
Kuffner, J., and LaValle, S.
\newblock 2000.
\newblock {RRT}-{C}onnect: An efficient approach to single-query path planning.
\newblock In {\em {IEEE} International Conference on Robotics and Automation},
  volume~2,  995--1001.

\bibitem[\protect\citeauthoryear{Phillips, Koenig, and
  Likhachev}{2014}]{phillips2014pastarse}
Phillips, M.; Koenig, S.; and Likhachev, M.
\newblock 2014.
\newblock Parallel {A}* for planning with time-consuming state expansions.
\newblock In {\em International Conference on Automated Planning and
  Scheduling}.

\bibitem[\protect\citeauthoryear{Pohl}{1970}]{pohl1970weightedastar}
Pohl, I.
\newblock 1970.
\newblock Heuristic search viewed as path finding in a graph.
\newblock {\em Artificial Intelligence} 1(3–4):193 -- 204.

\bibitem[\protect\citeauthoryear{Ramalingam and Reps}{1996}]{ramalingam1996}
Ramalingam, G., and Reps, T.
\newblock 1996.
\newblock An incremental algorithm for a generalization of the shortest-path
  problem.
\newblock {\em Journal of Algorithms} 21(2):267 -- 305.

\bibitem[\protect\citeauthoryear{Sint and de Champeaux}{1977}]{sint1977bhffa}
Sint, L., and de~Champeaux, D.
\newblock 1977.
\newblock An improved bidirectional heuristic search algorithm.
\newblock {\em J. ACM} 24(2):177--191.

\bibitem[\protect\citeauthoryear{Srinivasa \bgroup et al\mbox.\egroup
  }{2012}]{srinivasa2012herb20}
Srinivasa, S.; Berenson, D.; Cakmak, M.; Collet, A.; Dogar, M.; Dragan, A.;
  Knepper, R.; Niemueller, T.; Strabala, K.; Vande~Weghe, M.; and Ziegler, J.
\newblock 2012.
\newblock {HERB} 2.0: Lessons learned from developing a mobile manipulator for
  the home.
\newblock {\em Proceedings of the IEEE} 100(8):2410--2428.

\bibitem[\protect\citeauthoryear{Stentz}{1994}]{stentz1994dstar}
Stentz, A.
\newblock 1994.
\newblock Optimal and efficient path planning for partially-known environments.
\newblock In {\em {IEEE} International Conference on Robotics and Automation},
  volume~4,  3310--3317.

\bibitem[\protect\citeauthoryear{Stern, Puzis, and Felner}{2011}]{stern2011pts}
Stern, R.; Puzis, R.; and Felner, A.
\newblock 2011.
\newblock Potential search: A bounded-cost search algorithm.
\newblock In {\em International Conference on Automated Planning and
  Scheduling}.

\bibitem[\protect\citeauthoryear{{\c{S}}ucan, Moll, and
  Kavraki}{2012}]{sucan2012ompl}
{\c{S}}ucan, I.~A.; Moll, M.; and Kavraki, L.~E.
\newblock 2012.
\newblock The {O}pen {M}otion {P}lanning {L}ibrary.
\newblock {\em {IEEE} Robotics \& Automation Magazine} 19(4):72--82.
\newblock http://ompl.kavrakilab.org.

\bibitem[\protect\citeauthoryear{Yoshizumi, Miura, and
  Ishida}{2000}]{yoshizumi2000peastar}
Yoshizumi, T.; Miura, T.; and Ishida, T.
\newblock 2000.
\newblock A* with partial expansion for large branching factor problems.
\newblock In {\em {AAAI} Conference on Artificial Intelligence},  923--929.

\end{thebibliography}
\bibliographystyle{aaai}

\ifx\aaaiversion\undefined

\clearpage
\appendix
\section{Appendix: Proofs}
\label{sec:appendix-proofs}

\subsection{LazySP}

\begin{proof}[Proof of Theorem \ref{thm:lazy-optimality}]
Let $p^*$ be an optimal path w.r.t. $w$,
with $\ell^* = \mbox{len}(p^*, w)$.
Since $w_{\ms{est}}(e) \leq \epsilon \, w(e)$ and $\epsilon \geq 1$,
it follows that regardless of which edges are stored in $W_{\ms{eval}}$,
$w_{\ms{lazy}}(e) \leq \epsilon \, w(e)$,
and therefore
$\mbox{len}(p^*, w_{\ms{lazy}}) \leq \epsilon \, \ell^*$.
Now,
since the inner \textsc{ShortestPath} algorithm terminated with
$p_{\ms{ret}}$,
we know that
$\mbox{len}(p_{\ms{ret}}, w_{\ms{lazy}}) \leq \mbox{len}(p^*, w_{\ms{lazy}})$.
Further,
since the algorithm terminated with $p_{\ms{ret}}$,
each edge on $p_{\ms{ret}}$ has been evaluated;
therefore,
$\mbox{len}(p_{\ms{ret}}, w) = \mbox{len}(p_{\ms{ret}}, w_{\ms{lazy}})$.
Therefore,
$\mbox{len}(p_{\ms{ret}}, w) \leq \epsilon \, \ell^*$.
\end{proof}

\begin{proof}[Proof of Theorem \ref{thm:lazy-completeness}]
In this case,
the algorithm will evaluate at least unevaluated edge at
each iteration.
Since there are a finite number of edges,
eventually the algorithm will terminate.
\end{proof}

\subsection{A* Equivalence}

\begin{proof}[Proof of Invariant \ref{inv:astar-cundisc-popen}]
If $v$ is discovered, then it must either be on OPEN or CLOSED.
$v$ can be on CLOSED only after it has been expanded,
in which case $v'$ would be discovered (which it is not).
Therefore, $v$ must be on OPEN.
\end{proof}

\begin{proof}[Proof of Invariant \ref{inv:astar-wless-popen}]
Clearly the invariant holds at the beginning of the algorithm,
with only $v_{\ms{start}}$ on OPEN.
If the invariant were to no longer hold after some iteration,
then there must exist some pair of discovered vertices $v$ and $v'$
with $v$ on CLOSED and $g[v] + w(v,v') < g[v']$.
Since $v$ is on CLOSED,
it must have been expanded at some previous iteration,
immediately after which the inequality could not have held
because $g[v']$ is updated upon expansion of $v$.
Therefore,
the inequality must have newly held after some intervening iteration,
with $v$ remaining on CLOSED.
Since the values $g$ are monotonically non-increasing and $w$ is fixed,
this implies that $g[v]$ must have been updated (lower).
However,
if this had happened,
then $v$ would have been removed from CLOSED and placed on OPEN.
This contradiction implies that the invariant holds at every iteration.
\end{proof}

\begin{proof}[Proof of Theorem \ref{thm:astar-equiv-from-lazy}]
Consider path $p^*_{\ms{lazy}}$ with length $\ell^*_{\ms{lazy}}$
yielding frontier vertex $v_{\ms{frontier}}$ via \textsc{SelectExpand}.
Construct a vertex sequence $s$ as follows.
Initialize $s$ with the vertices on $p^*_{\ms{lazy}}$
from $v_{\ms{start}}$ to $v_{\ms{frontier}}$, inclusive.
Let $N$ be the number of consecutive vertices at the start of $s$
for which $f(v) = \ell^*_{\ms{lazy}}$
(Note that the first vertex on $p^*_{\ms{lazy}}$, $v_{\ms{start}}$,
must have $f(v_{\ms{start}}) = \ell^*_{\ms{lazy}}$,
so $N \geq 1$.)
Remove from the start of $s$ the first $N-1$ vertices.
Note that at most the first vertex on $s$ has
$f(v) = \ell^*_{\ms{lazy}}$,
and the last vertex on $s$ must be $v_{\ms{frontier}}$.

Now we show that each vertex in this sequence $s$,
considered by A* in turn,
exists on OPEN with minimal $f$-value.
Iteratively consider the following procedure for sequence $s$.
Throughout,
we know that there must not be any vertex with
$f(v) < \ell^*_{\ms{lazy}}$;
that would imply that a different path through $v_b$ shorter than
$\ell^*_{\ms{lazy}}$ exists,
in which case $p^*_{\ms{lazy}}$ could not have been chosen.

If the sequence has length $>1$,
then consider the first two vertices on $s$, $v_a$ and $v_b$.
By construction,
$f(v_a) = \ell^*_{\ms{lazy}}$
and 
$f(v_b) \neq \ell^*_{\ms{lazy}}$.
In fact, from above
we know that $f(v_b) > \ell^*_{\ms{lazy}}$.
Therefore,
we have that $f(v_a) < f(v_b)$,
therefore and $g[v_a] + w(v_a,v_b) < g[v_b]$.
By Invariant~\ref{inv:astar-wless-popen},
$v_a$ must be on OPEN,
and with $f(v_a) = \ell^*_{\ms{lazy}}$,
it can therefore be considered by A*.
After it is expanded,
$f(v_b) = \ell^*_{\ms{lazy}}$,
and we can repeat the above procedure
with the sequence formed by removing the $v_a$ from $s$.

If instead the sequence has length $1$,
then it must be exactly $(v_{\ms{frontier}})$,
with $f(v_{\ms{frontier}}) = \ell^*_{\ms{lazy}}$.
Since the edge after $f(v_{\ms{frontier}})$ is not
evaluated,
then by Invariant~\ref{inv:astar-cundisc-popen},
$v_{\ms{frontier}}$ must be on OPEN,
and will therefore be expanded next.
\end{proof}

\begin{proof}[Proof of Theorem \ref{thm:astar-equiv-to-lazy}]
Given that all vertices in $s_{\ms{candidate}}$ besides the last
are re-expansions,
they can be expanded with no edge evaluations.
Once the last vertex,
$v_{\ms{frontier}}$,
is to be expanded by A*,
suppose it has $f$-value $\ell$.

First,
we will show that there exists a path with length $\ell$ w.r.t.
$w_{\ms{lazy}}$
wherein all edges before $v_{\ms{frontier}}$ have been evaluated,
and the first edge after $v_{\ms{frontier}}$ has not.
Let $p_a$ be a shortest path from $v_{\ms{start}}$
to $v_{\ms{frontier}}$ consisting of only evaluated edges.
The length of this $p_a$ must be equal to $g[v_{\ms{frontier}}]$;
if it were not,
there would be some previous vertex on $p_a$ with lower $f$-value
than $v_{\ms{frontier}}$,
which would necessarily have been expanded first.
Let $p_b$ be the a shortest path from $v_{\ms{frontier}}$
to $v_{\ms{goal}}$.
The length of $p_b$ must be $h_{\ms{lazy}}(v_{\ms{frontier}})$
by definition.
Therefore, the path $(p_a, p_b)$ must have length $\ell$,
and since $v_{\ms{frontier}}$ is a new expansion,
the first edge on $p_b$ must be unevaluated.

Second,
we will show that there does not exist any path shorter than $\ell$
w.r.t. $w_{\ms{lazy}}$.
Suppose $p'$ were such a path, with length $\ell' < \ell$.
Clearly, $v_{\ms{start}}$ would have $f$-value $\ell'$ (although
it may not be on OPEN).
Consider each pair of vertices $(v_a, v_b)$ along $p'$ in turn.
In each case,
if $v_b$ were either undiscovered,
or if $g[v_a] + w(v_a, v_b) < g[v_b]$,
then $v_a$ would be on OPEN
(via Invariants \ref{thm:lazy-completeness}
and \ref{inv:astar-wless-popen}, respectively)
with $f(v_a) = \ell'$,
and would therefore have been expanded before $v_{\ms{frontier}}$.
Otherwise,
we know that $f(v_b) = \ell'$,
and we can continue to the next pair on $p'$.
\end{proof}

\subsection{LWA* Equivalence}

\begin{proof}[Proof of Invariant \ref{inv:lwastar}]
Clearly the invariant holds at the beginning of the algorithm
with only $g[v_{\ms{start}}] = 0$,
since the inequality holds only for the out-edges of $v_{\ms{start}}$,
with $v_{\ms{start}}$ on $Q_v$.
Consider each subsequent iteration.
If a vertex $v$ is popped from $Q_v$,
then this may invalidate the invariant for all successors of $v$;
however,
since all out-edges are immediately added to $Q_e$,
the invariant must hold.
Consider instead if an edge $(v, v')$ which satisfies the inequality
is popped from $Q_e$.
Due to the inequality,
we know that $g[v']$ will be recalculated as
$g[v'] = g[v] + w(v,v')$,
so that the inequality is no longer satisfied for edge $(v,v')$.
However,
reducing the value $g[v']$ may introduce satisfied inequalities across
subsequent out-edges of $v'$,
but since $v'$ is added to $Q_v$,
the invariant continues to hold.
\end{proof}

\begin{proof}[Proof of Theorem \ref{thm:lwastar-equiv-from-lazy}]
In the first component of the equivalence,
we will show that for any path $p$ minimizing $w_{\ms{lazy}}$
allowable to LazySP-Forward,
with $(v_a, v_b)$ the first unevaluated edge on $p$,
there exists a sequence of vertices and edges on
$Q_v$ and $Q_e$ allowable to LWA*
such that edge $(v_a, v_b)$ is the first to be newly evaluated.
Let the length of $p$ w.r.t. $w_{\ms{lazy}}$ be $\ell$.'

We will first show that no vertex on $Q_v$ or edge on $Q_e$
can have $f(\cdot) < \ell$.
Suppose such a vertex $v$, or edge $e$ with source vertex $v$,
exists.
Then $g[v] + h(v) < \ell$,
and there must be some path $p'$ consisting of an evaluated segment
from $v_{\ms{start}}$ to $v$ of length $g[v]$,
followed a segment from $v$ to $v_{\ms{goal}}$ of length $h(v)$.
But then this path should have been chosen by LazySP.

Next, we will show a procedure for generating an allowable
sequence for LWA*.
We will iteratively consider a sequence of path segments,
starting with the segment from $v_{\ms{start}}$ to $v_a$,
and becoming progressively shorter at each iteration by removing the
first vertex and edge on the path.
We will show that the first vertex on each segment $v_f$
has $g[v_f] = \ell - h(v_f)$.
By definition,
this is true of the first such segment, since $g[v_{\ms{start}}] = 0$.
For each but the last such segment,
consider the first edge, $(v_f, v_s)$.
If $v_s$ has the correct $g[\cdot]$,
we can continue to the next segment immediately.
Otherwise,
either $v_f$ is on $Q_v$ or $(v_f, v_s)$ is on $Q_e$ by
Invariant~\ref{inv:lwastar}.
If the former is true,
then $v_f$ can be popped from $Q_v$ with $f = \ell$,
thereby adding $(v_f, v_s)$ to $Q_e$.
Then,
$(v_f, v_s)$ can be popped from $Q_e$ with $f = \ell$,
resulting in $g[v_s] = \ell - h(v_s)$.
We can then move on to the next segment.

At the end of this process,
we have the trivial segment $(v_a)$,
with $g[v_a] = \ell - h(v_a)$.
If $v_a$ is on $Q_v$, then pop it (with $f(v_a) = \ell$),
placing $e_{ab}$ on $Q_e$;
otherwise, since $e_{ab}$ is unevaluated,
it must already be on $Q_e$.
Since $f(e_{ab}) = \ell$, we can pop and evaluate it.
\end{proof}

\begin{proof}[Proof of Theorem \ref{thm:lwastar-equiv-to-lazy}]
Given that all vertices in $s_{\ms{candidate}}$ entail no edge evaluations,
and all edges therein are re-expansions,
they can be considered with no edge evaluations.
Once the last edge $e_{xy}$ is to be expanded by LWA*,
suppose it has $f$-value $\ell$.

First,
we will show that there exists a path with length $\ell$ w.r.t.
$w_{\ms{lazy}}$
which traverses unevaluated edge $e_{xy}$
wherein all edges before $v_x$ have been evaluated.
Let $p_x$ be a shortest path segment from $v_{\ms{start}}$
to $v_x$ consisting of only evaluated edges.
The length of $p_x$ must be equal to $g[v_x]$;
if it were not,
there would be some previous vertex on $p_x$ with lower $f$-value
than $v_x$,
which would necessarily have been expanded first.
Let $p_y$ be the a shortest path from $v_y$
to $v_{\ms{goal}}$.
The length of $p_y$ must be $h_{\ms{lazy}}(v_y)$
by definition.
Therefore, the path $(p_x, e_{xy}, p_y)$ must have length $\ell$.

Second,
we will show that there does not exist any path shorter than $\ell$
w.r.t. $w_{\ms{lazy}}$.
Suppose $p'$ were such a path, with length $\ell' < \ell$,
and with first unevaluated edge $e'_{xy}$.
Clearly, $v_{\ms{start}}$ has
$g[v_{\ms{start}}] = \ell' - h(v_{\ms{start}})$.
Consider each evaluated edge $e'_{ab}$ along $p'$ in turn.
In each case,
if $g[v'_b] \neq \ell' - h(v'_b)$,
then either $v'_a$ or $e'_{ab}$ would be on $Q_v$ or $Q_e$
with $f(\cdot) = \ell'$,
and would therefore be expanded before $e_{xy}$.
Therefore,
$e'_{xy}$ would then be popped from $Q_e$ with $f(e'_{xy}) = \ell'$,
and it would have been evaluated before $e_{xy}$ with $f(e_{xy}) = \ell$.
\end{proof}

\section{Appendix: Timing Results}
\label{sec:appendix-timing}

We include an accounting of the cumulative computation time taken by
each component of LazySP for each of the seven selectors
for each problem class (Figure~\ref{fig:table-timing-results}).

\balance

\begin{figure*}[t]
   \centering
   {\small%
   \setlength{\tabcolsep}{0.06cm}%
   \begin{tabular}{l@{\hskip 9pt}rc@{\hskip 0pt}r@{\hskip 9pt}rc@{\hskip 0pt}r@{\hskip 9pt}rc@{\hskip 0pt}r@{\hskip 9pt}rc@{\hskip 0pt}r@{\hskip 9pt}rc@{\hskip 0pt}r@{\hskip 9pt}rc@{\hskip 0pt}r@{\hskip 9pt}rc@{\hskip 0pt}r}
      \toprule
         & \multicolumn{3}{c}{Expand} & \multicolumn{3}{c}{Forward} & \multicolumn{3}{c}{Reverse} & \multicolumn{3}{c}{Alternate}
         & \multicolumn{3}{c}{Bisect} & \multicolumn{3}{c}{WeightSamp} & \multicolumn{3}{c}{Partition} \\
      \midrule
      \addlinespace[0.5em]
      PartConn & & & & & & & & & & & & & & & & & & & & & \\
      \;\;\emph{total (ms)}    &\bf1.22 &$\pm$&   0.04 &    1.96 &$\pm$&  0.06 &    1.86 &$\pm$&  0.06 &\bf1.20 &$\pm$& 0.03 &  2.41 &$\pm$& 0.06 & 4807.19 &$\pm$& 135.22 &   15.81 &$\pm$&  0.16 \\
      \;\;\emph{sel-init (ms)} & --\;\; &     &        &  --\;\; &     &       &  --\;\; &     &       & --\;\; &     &      &--\;\; &     &      &  --\;\; &     &        &   12.49 &$\pm$&  0.11 \\
      \;\;\emph{online (ms)}   &\bf1.22 &$\pm$&   0.04 &    1.96 &$\pm$&  0.06 &    1.86 &$\pm$&  0.06 &\bf1.20 &$\pm$& 0.03 &  2.41 &$\pm$& 0.06 & 4807.19 &$\pm$& 135.22 &    3.32 &$\pm$&  0.10 \\
      \;\;\emph{search (ms)}   &   0.48 &$\pm$&   0.01 &    1.12 &$\pm$&  0.03 &    1.05 &$\pm$&  0.03 &   0.68 &$\pm$& 0.02 &  1.38 &$\pm$& 0.04 &    0.70 &$\pm$&   0.02 &    0.68 &$\pm$&  0.02 \\
      \;\;\emph{sel (ms)}      &   0.02 &$\pm$&   0.00 &    0.01 &$\pm$&  0.00 &    0.01 &$\pm$&  0.00 &   0.01 &$\pm$& 0.00 &  0.03 &$\pm$& 0.00 & 4805.64 &$\pm$& 135.18 &    2.07 &$\pm$&  0.06 \\
      \;\;\emph{eval (ms)}     & --\;\; &     &        &  --\;\; &     &       &  --\;\; &     &       & --\;\; &     &      &--\;\; &     &      &  --\;\; &     &        &  --\;\; &     &       \\
      \;\;\emph{eval (edges)}  &  87.10 &$\pm$&   2.39 &   35.86 &$\pm$&  1.04 &   34.84 &$\pm$&  1.04 &  22.23 &$\pm$& 0.60 & 44.81 &$\pm$& 1.11 &\bf20.66 &$\pm$&   0.57 &\bf20.39 &$\pm$&  0.56 \\
      \addlinespace[0.5em]
      UnitSquare & & & & & & & & & & & & & & & & & & & & & \\
      \;\;\emph{total (ms)}    &\bf0.91 &$\pm$&   0.03 &    1.47 &$\pm$&  0.06 &    1.49 &$\pm$&  0.06 &\bf0.94 &$\pm$& 0.03 &  1.71 &$\pm$& 0.04 & 3864.95 &$\pm$& 117.66 &   15.13 &$\pm$&  0.14 \\
      \;\;\emph{sel-init (ms)} & --\;\; &     &        &  --\;\; &     &       &  --\;\; &     &       & --\;\; &     &      &--\;\; &     &      &  --\;\; &     &        &   13.41 &$\pm$&  0.12 \\
      \;\;\emph{online (ms)}   &\bf0.91 &$\pm$&   0.03 &    1.47 &$\pm$&  0.06 &    1.49 &$\pm$&  0.06 &\bf0.94 &$\pm$& 0.03 &  1.71 &$\pm$& 0.04 & 3864.95 &$\pm$& 117.66 &    1.72 &$\pm$&  0.06 \\
      \;\;\emph{search (ms)}   &   0.35 &$\pm$&   0.01 &    0.79 &$\pm$&  0.03 &    0.82 &$\pm$&  0.03 &   0.51 &$\pm$& 0.02 &  0.92 &$\pm$& 0.02 &    0.75 &$\pm$&   0.02 &    0.45 &$\pm$&  0.01 \\
      \;\;\emph{sel (ms)}      &   0.01 &$\pm$&   0.00 &    0.01 &$\pm$&  0.00 &    0.01 &$\pm$&  0.00 &   0.01 &$\pm$& 0.00 &  0.02 &$\pm$& 0.00 & 3863.49 &$\pm$& 117.62 &    0.87 &$\pm$&  0.03 \\
      \;\;\emph{eval (ms)}     & --\;\; &     &        &  --\;\; &     &       &  --\;\; &     &       & --\;\; &     &      &--\;\; &     &      &  --\;\; &     &        &  --\;\; &     &       \\
      \;\;\emph{eval (edges)}  &  69.21 &$\pm$&   2.55 &   27.29 &$\pm$&  1.03 &   27.69 &$\pm$&  1.02 &  17.82 &$\pm$& 0.60 & 32.62 &$\pm$& 0.72 &   15.58 &$\pm$&   0.47 &\bf14.08 &$\pm$&  0.46 \\
      \addlinespace[0.5em]
      ArmPlan (avg) & & & & & & & & & & & & & & & & & & & & & \\
      \;\;\emph{total (s)}    &  269.82 &$\pm$&  17.95 &\bf 5.90 &$\pm$&  0.46 &    8.22 &$\pm$&  0.53 &\bf5.96 &$\pm$& 0.31 &  7.34 &$\pm$& 0.43 & 3402.21 &$\pm$& 172.20 &  496.57 &$\pm$&  5.53 \\
      \;\;\emph{sel-init (s)} &  --\;\; &     &        &  --\;\; &     &       &  --\;\; &     &       & --\;\; &     &      &--\;\; &     &      &  --\;\; &     &        &  490.77 &$\pm$&  5.51 \\
      \;\;\emph{online (s)}   &  269.82 &$\pm$&  17.95 &\bf 5.90 &$\pm$&  0.46 &    8.22 &$\pm$&  0.53 &\bf5.96 &$\pm$& 0.31 &  7.34 &$\pm$& 0.43 & 3402.21 &$\pm$& 172.20 &\bf 5.80 &$\pm$&  0.28 \\
      \;\;\emph{search (s)}   &    0.02 &$\pm$&   0.00 &    0.02 &$\pm$&  0.00 &    0.02 &$\pm$&  0.00 &   0.02 &$\pm$& 0.00 &  0.02 &$\pm$& 0.00 &    0.02 &$\pm$&   0.00 &    0.04 &$\pm$&  0.00 \\
      \;\;\emph{sel (s)}      &    0.00 &$\pm$&   0.00 &    0.00 &$\pm$&  0.00 &    0.00 &$\pm$&  0.00 &   0.00 &$\pm$& 0.00 &  0.00 &$\pm$& 0.00 & 3392.76 &$\pm$& 171.74 &    1.54 &$\pm$&  0.09 \\
      \;\;\emph{eval (s)}     &  269.78 &$\pm$&  17.95 &    5.87 &$\pm$&  0.45 &    8.20 &$\pm$&  0.52 &   5.94 &$\pm$& 0.31 &  7.31 &$\pm$& 0.43 &    9.39 &$\pm$&   0.57 &    4.21 &$\pm$&  0.22 \\
      \;\;\emph{eval (edges)} &  949.05 &$\pm$&  63.46 &   63.62 &$\pm$&  4.15 &   74.94 &$\pm$&  5.07 &  55.48 &$\pm$& 2.95 & 68.01 &$\pm$& 3.86 &   56.93 &$\pm$&   3.37 &\bf48.07 &$\pm$&  2.44 \\
      \addlinespace[0.25em]
      ArmPlan1 & & & & & & & & & & & & & & & & & & & & & \\
      \;\;\emph{total (s)}    &  109.09 &$\pm$&  14.15 &\bf 4.81 &$\pm$&  0.49 &   14.81 &$\pm$&  1.45 &   7.03 &$\pm$& 0.63 &  7.91 &$\pm$& 0.70 & 3375.35 &$\pm$& 319.81 &  496.74 &$\pm$&  8.22 \\
      \;\;\emph{sel-init (s)} &  --\;\; &     &        &  --\;\; &     &       &  --\;\; &     &       & --\;\; &     &      &--\;\; &     &      &  --\;\; &     &        &  489.49 &$\pm$&  8.18 \\
      \;\;\emph{online (s)}   &  109.09 &$\pm$&  14.15 &\bf 4.81 &$\pm$&  0.49 &   14.81 &$\pm$&  1.45 &   7.03 &$\pm$& 0.63 &  7.91 &$\pm$& 0.70 & 3375.35 &$\pm$& 319.81 &    7.25 &$\pm$&  0.66 \\
      \;\;\emph{search (s)}   &    0.02 &$\pm$&   0.00 &    0.02 &$\pm$&  0.00 &    0.03 &$\pm$&  0.00 &   0.02 &$\pm$& 0.00 &  0.02 &$\pm$& 0.00 &    0.02 &$\pm$&   0.00 &    0.04 &$\pm$&  0.00 \\
      \;\;\emph{sel (s)}      &    0.00 &$\pm$&   0.00 &    0.00 &$\pm$&  0.00 &    0.00 &$\pm$&  0.00 &   0.00 &$\pm$& 0.00 &  0.00 &$\pm$& 0.00 & 3358.82 &$\pm$& 318.17 &    1.61 &$\pm$&  0.16 \\
      \;\;\emph{eval (s)}     &  109.07 &$\pm$&  14.15 &    4.78 &$\pm$&  0.49 &   14.77 &$\pm$&  1.44 &   7.01 &$\pm$& 0.63 &  7.88 &$\pm$& 0.70 &   16.47 &$\pm$&   1.68 &    5.59 &$\pm$&  0.51 \\
      \;\;\emph{eval (edges)} &  344.74 &$\pm$&  39.63 &\bf49.72 &$\pm$&  4.25 &   95.58 &$\pm$&  9.67 &  59.44 &$\pm$& 5.06 & 58.90 &$\pm$& 4.74 &   73.72 &$\pm$&   7.63 &\bf50.66 &$\pm$&  4.43 \\
      \addlinespace[0.25em]
      ArmPlan2 & & & & & & & & & & & & & & & & & & & & & \\
      \;\;\emph{total (s)}    &  166.19 &$\pm$&   9.29 &\bf 3.27 &$\pm$&  0.25 &    7.36 &$\pm$&  0.69 &   5.95 &$\pm$& 0.52 &  5.63 &$\pm$& 0.45 & 4758.04 &$\pm$& 407.56 &  495.21 &$\pm$& 12.65 \\
      \;\;\emph{sel-init (s)} &  -    - &     &        &  --\;\; &     &       &  --\;\; &     &       & --\;\; &     &      &--\;\; &     &      &  --\;\; &     &        &  489.22 &$\pm$& 12.64 \\
      \;\;\emph{online (s)}   &  166.19 &$\pm$&   9.29 &\bf 3.27 &$\pm$&  0.25 &    7.36 &$\pm$&  0.69 &   5.95 &$\pm$& 0.52 &  5.63 &$\pm$& 0.45 & 4758.04 &$\pm$& 407.56 &    5.99 &$\pm$&  0.48 \\
      \;\;\emph{search (s)}   &    0.01 &$\pm$&   0.00 &    0.01 &$\pm$&  0.00 &    0.02 &$\pm$&  0.00 &   0.01 &$\pm$& 0.00 &  0.01 &$\pm$& 0.00 &    0.02 &$\pm$&   0.00 &    0.03 &$\pm$&  0.00 \\
      \;\;\emph{sel (s)}      &    0.00 &$\pm$&   0.00 &    0.00 &$\pm$&  0.00 &    0.00 &$\pm$&  0.00 &   0.00 &$\pm$& 0.00 &  0.00 &$\pm$& 0.00 & 4750.16 &$\pm$& 406.98 &    2.03 &$\pm$&  0.22 \\
      \;\;\emph{eval (s)}     &  166.17 &$\pm$&   9.28 &\bf 3.26 &$\pm$&  0.25 &    7.34 &$\pm$&  0.69 &   5.93 &$\pm$& 0.52 &  5.61 &$\pm$& 0.45 &    7.82 &$\pm$&   0.61 &    3.91 &$\pm$&  0.27 \\
      \;\;\emph{eval (edges)} &  657.02 &$\pm$&  29.24 &\bf62.24 &$\pm$&  6.12 &   98.54 &$\pm$& 10.89 &  69.96 &$\pm$& 6.98 & 75.88 &$\pm$& 7.47 &\bf66.24 &$\pm$&   6.36 &\bf62.16 &$\pm$&  6.10 \\
      \addlinespace[0.25em]
      ArmPlan3 & & & & & & & & & & & & & & & & & & & & & \\
      \;\;\emph{total (s)}    &  534.16 &$\pm$&  55.64 &    9.61 &$\pm$&  1.33 &\bf 2.50 &$\pm$&  0.23 &   4.91 &$\pm$& 0.56 &  8.47 &$\pm$& 0.99 & 2073.23 &$\pm$& 198.75 &  497.76 &$\pm$& 10.27 \\
      \;\;\emph{sel-init (s)} &  --\;\; &     &        &  --\;\; &     &       &  --\;\; &     &       & --\;\; &     &      &--\;\; &     &      &  --\;\; &     &        &  493.59 &$\pm$& 10.21 \\
      \;\;\emph{online (s)}   &  534.16 &$\pm$&  55.64 &    9.61 &$\pm$&  1.33 &\bf 2.50 &$\pm$&  0.23 &   4.91 &$\pm$& 0.56 &  8.47 &$\pm$& 0.99 & 2073.23 &$\pm$& 198.75 &    4.17 &$\pm$&  0.43 \\
      \;\;\emph{search (s)}   &    0.02 &$\pm$&   0.01 &    0.02 &$\pm$&  0.00 &    0.02 &$\pm$&  0.00 &   0.02 &$\pm$& 0.00 &  0.02 &$\pm$& 0.00 &    0.03 &$\pm$&   0.01 &    0.04 &$\pm$&  0.00 \\
      \;\;\emph{sel (s)}      &    0.00 &$\pm$&   0.00 &    0.00 &$\pm$&  0.00 &    0.00 &$\pm$&  0.00 &   0.00 &$\pm$& 0.00 &  0.00 &$\pm$& 0.00 & 2069.29 &$\pm$& 198.53 &    0.98 &$\pm$&  0.13 \\
      \;\;\emph{eval (s)}     &  534.10 &$\pm$&  55.63 &    9.58 &$\pm$&  1.33 &    2.48 &$\pm$&  0.23 &   4.89 &$\pm$& 0.56 &  8.44 &$\pm$& 0.99 &    3.90 &$\pm$&   0.31 &    3.15 &$\pm$&  0.32 \\
      \;\;\emph{eval (edges)} & 1845.38 &$\pm$& 195.57 &   78.90 &$\pm$& 10.36 &\bf30.70 &$\pm$&  3.62 &  37.04 &$\pm$& 4.59 & 69.26 &$\pm$& 7.97 &\bf30.82 &$\pm$&   3.60 &\bf31.38 &$\pm$&  3.80 \\
      \addlinespace[0.25em]
      \bottomrule
   \end{tabular}%
   }%
   \caption{
      Detailed timing results for each selector.
      The actual edge weights for the illustrative
      PartConn and UnitSquare problems were pre-computed,
      and therefore their timings are not included.
      The Partition selector requires initialization of the $Z$-values
      (\ref{eqn:partitionfn}) for the graph using only the estimated
      edge weights.
      Since this is not particular to either the actual edge weights
      (e.g. from the obstacle distribution)
      or the start/goal vertices from a particular instance,
      this initialization (\emph{sel-init}) is considered separately.
      The online running time (\emph{online}) is broken into LazySP's
      three primary steps: the inner search (\emph{search}),
      invoking the edge selector (\emph{sel}),
      and evaluating edges (\emph{eval}).
      We also show the number of edges evaluated.}
   \label{fig:table-timing-results}
\end{figure*}

\fi 

\end{document}